\documentclass[12pt,reqno,hidelinks]{amsart}
\usepackage{amssymb}
\usepackage{amsmath, amssymb, verbatim, url, mathrsfs}
\usepackage{mathrsfs}
\usepackage{mathtools}
\usepackage{verbatim}
\usepackage{amsthm}
\usepackage{framed}
\usepackage{cite}
\usepackage{wasysym}
\usepackage{upgreek}
\usepackage{color}
\usepackage[dvipsnames]{xcolor}
\usepackage{tensor}
\usepackage{accents}
\usepackage{dsfont}
\usepackage{hyperref}
\usepackage{enumerate}
\usepackage[normalem]{ulem}
\usepackage{longtable}
\usepackage{mathtools}

\numberwithin{equation}{section}

\newtheorem{proposition}{Proposition}[subsection]
\newtheorem{lemma}[proposition]{Lemma}
\newtheorem{corollary}[proposition]{Corollary}
\newtheorem{theorem}{Theorem}[section]

\newtheorem{claim}{Claim}[section]

\theoremstyle{definition}
\newtheorem{definition}{Definition}[subsection]
\newtheorem{remark}{Remark}[subsection]

\newcommand{\lie}{\mathcal{L}}

\newcommand{\nbar}{\bar \nabla}

\newcommand{\he}{\hat{e}}
\newcommand{\ha}{\hat{a}}
\newcommand{\hb}{\hat{b}}
\newcommand{\hc}{\hat{c}}
\newcommand{\hd}{\hat{d}}
\newcommand{\vertiiii}[1]{{\left\vert\kern-0.25ex\left\vert\kern-0.25ex\left\vert\kern-0.25ex\left\vert #1 \right\vert\kern-0.25ex\right\vert\kern-0.25ex\right\vert\kern-0.25ex\right\vert}}
\newcommand{\vertiii}[1]{{\left\vert\kern-0.25ex\left\vert\kern-0.25ex\left\vert #1 \right\vert\kern-0.25ex\right\vert\kern-0.25ex\right\vert}}

\newcommand{\tE}{{\mathcal{E}}}

\newcommand{\nnb}{\nonumber}
\newcommand{\ulg}{\underline{g}}
\newcommand{\ulh}{\underline{h}}
\newcommand{\ulR}{\underline{R}} 
 
\newcommand{\Rbb}{\mathbb{R}}
\newcommand{\Zbb}{\mathbb{Z}}

\newcommand{\Pbb}{\mathbb{P}}

\newcommand{\udl}{\underline}

\newcommand{\udn}[1]{\underline{\nabla}_{#1}}

\newcommand{\del}[1]{{\partial_{#1}}}

\newcommand{\AND}{{\quad\text{and}\quad}}

\newcommand{\Li}{L^\infty}
\newcommand{\la}{\langle}
\newcommand{\ra}{\rangle}

\newcommand{\A}{\mathcal{A}}

\newcommand{\al}[2]{
	\begin{align}\label{E:#1}
		#2
	\end{align}
}

\newcommand{\gat}[1]{
	\begin{gather}
		#1
	\end{gather}
}
\newcommand{\als}[1]{
	\begin{align*}
		#1
	\end{align*}
}
\newcommand{\p}[1]{
	\begin{pmatrix}
		#1
	\end{pmatrix}
}

\newcommand{\ts}{\tensor}

\newcommand{\di}{\mathrm{d}} 
\newcommand{\tr}{\text{tr}}

\newcommand{\nb}{\underline{\nabla}}

\newcommand{\be}{\begin{equation}}
	\newcommand{\ee}{\end{equation}}

\allowdisplaybreaks

\DeclareMathOperator{\diag}{diag}

\DeclareMathOperator{\Ima}{Im}

\begin{document}

	\title[Local Einstein--Yang--Mills system in the temporal gauge]{A new symmetric hyperbolic formulation and the local Cauchy problem for the Einstein--Yang--Mills system in the temporal gauge}

	\author{Chao Liu}
	\address[Chao Liu]{Center for Mathematical Sciences and School of Mathematics and Statistics, Huazhong University of Science and Technology, Wuhan 430074, Hubei Province, China.}
	\email{chao.liu.math@foxmail.com}
	
	\author{Jinhua Wang}
\address[Jinhua Wang]{School of Mathematical Science, Xiamen University,
	Xiamen 361005, Fujian Province, China. }
\email{wangjinhua@xmu.edu.cn}

\begin{abstract}
	\textit{Motivated} by the future stability problem of solutions of the Einstein--Yang--Mills (EYM) system \textit{with arbitrary dimension}, we \textit{aim} to $(1)$ construct a tensorial symmetric hyperbolic formulation for the $(n+1)$-dimensional EYM system in the \textit{temporal gauge}; $(2)$ establish the local well-posedness for the Cauchy problem of EYM equations in the \textit{temporal gauge} using this tensorial symmetric hyperbolic system.  By introducing certain auxiliary variables, we extend essentially the $(n+1)$-dimensional Yang--Mills system to a tensorial symmetric hyperbolic system. On the contrary, this symmetric hyperbolic system with data satisfying some constraints (extending the Yang--Mills constraints) reduces to the Yang--Mills system. Consequently, an equivalence between the EYM and the tensorial symmetric hyperbolic system with a class of specific data set is concluded. Furthermore, a general symmetric hyperbolic system over tensor bundles is studied, with which, we conclude the local well-posedness of the EYM system. 	
	It turns out the idea of symmetric hyperbolic formulation of the Yang--Mills field is very useful in prompting a tensorial Fuchsian formalism and  proving the future stability for the EYM system with arbitrary dimension (i.e., this new symmetric hyperbolic formulation of EYM manifests well behaved lower order terms for long time evolution), see our companion article \cite{LW2021a} with Todd A. Oliynyk.
	


		\vspace{2mm}

\end{abstract}

\maketitle

\setcounter{tocdepth}{2}

\pagenumbering{roman} \pagenumbering{arabic}

\section{Introduction}
In general relativity, the Einstein equations describe the developments of the spacetime and interactions with the source fields. Meanwhile, the Yang--Mills theory, as the foundation of the standard model of particle physics, is a gauge theory with a non-Abelian local symmetry group $G$ and the Yang--Mills field can serve as a (matter) \textit{field source} of the Einstein equations. Particularly, in the case that the group $G$ is an Abelian group $U(1)$, it reduces to the electromagnetic theory in terms of the Maxwell equations.

It is physically and mathematically important to understand the  local and  global behaviours of solutions of the fully nonlinear Einstein--Yang--Mills (EYM, for short) system. This article and the companion one \cite{LW2021a} serve these purposes of the local and global ones, respectively. Although the local result of the EYM system is relatively simple and has been studied in some contexts by Choquet-Bruhat and Friedrich, due to the needs of developing a framework for the global results with \textit{arbitrary dimension}, we, in this article, \textit{aim} to
$(1)$ propose a new tensorial symmetric hyperbolic formulation for the $(n+1)$-dimensional EYM system in \textit{temporal gauge} and $(2)$ establish the local well-posedness of the Cauchy problem of EYM equations in the \textit{temporal gauge}. Let us briefly recall some previous works and state our main motivations.

The \textit{main motivation} of this article is to construct a symmetric hyperbolic system of the EYM equations in the temporal gauge and ensure that this formulation manifests nice behaviour for the global problem. The local well-posedness for the Einstein--Yang--Mills  equations with any dimension has been implied by Choquet-Bruhat \cite{Choquet-Bruhat2009,ChoquetBruhat1992} through the obvious hyperbolicity in the \textit{Lorentz gauge} (similar to the Einstein--Maxwell system), while she pointed out ``the global results may be quite different'', since it requires well behaved lower order terms for the global problem and the formulation in the Lorentz gauge simply may not serve for it. We attempt, as the \textit{first motivation}, to find a formulation of the EYM which may be proper for the global problem as well.
Although Friedrich \cite{Friedrich-hyperbolicity, Friedrich1991} was able to show the hyperbolicity of the Einstein equations and other gauge field such as Yang--Mills field, and further yielded the global existence of solutions of the \textit{four dimensional} Einstein--Maxwell--Yang--Mills system, the conformal method he applied to this work has some difficulties for higher dimensional cases. The formulation that we try to construct initially in this paper and to be developed for global picture in \cite{LW2021a} will be suitable to work for any dimension $n+1\geq4$. This is the \textit{second motivation} for pursuing this new formulation.

To begin with, we briefly recall the Einstein and Yang--Mills theory (for more details, please refer to, for instance, \cite{Ghanem2014,Choquet-Bruhat2009,Ginibre1982,Segal1979,ChoquetBruhat1991,Eardley1982,Ginibre1981} and references therein). Suppose $(\mathcal{M}^{n+1},  g_{a b})$ is a connected Lorentzian manifold.		
Let $G$ be an $(n+1)$-dimensional connected Lie group with Lie algebra $\mathcal{G}$. Throughout we  assume
that $\mathcal{G}$ admits an Ad-invariant positive scalar product denoted by a dot ``$ \cdot $''.
A connection $A$ on a $G$-principal bundle over $\mathcal{M}^{n+1}$ can be expressed as a $\mathcal{G}$-valued $1$-form
and the curvature $ F$ is the $\mathcal{G}$-valued $2$-form given in terms of $A$ by (with $[\cdot,\cdot]$ the Lie bracket in
$\mathcal{G}$ and $\nabla$ the covariant derivative associated to ${g}_{a b}$),
\be\label{def-F}
F = \di A + [A, A ], \quad \text{i.e.,} \quad  F_{ab} = \nabla_{a} A_b-\nabla_{b} A_a + [A_a, A_b ].
\ee  
The equation of motion of the free Yang--Mills field is governed by the conservation law,
\begin{align}
	D^a  F_{a c}  
	& =0,  \label{eq-maxwell-div-o0}
\end{align}
where the symbol $ D =  \nabla + [ A, \cdot]$ denotes the gauge covariant derivative of a $\mathcal{G}$-valued tensor.
In addition, by the construction \eqref{def-F}, the Yang-Mills \textit{Bianchi identities} hold true
\be
D_{[a}  F_{b c]} = 0.  \label{eq-YM-bianchi-0}
\ee

The Einstein--Yang--Mills system with a cosmological constant $\Lambda\in \Rbb$, is given by 
\begin{align}
	G_{a b} + \Lambda  g_{a b}& =  \mathcal{T}_{a b}, \label{eq-einstein} \\
	D^a  F_{a b} & =0,  \label{eq-YM-div-0}  \\
	D_{[a}  F_{b c]}& = 0.  \label{eq-YM-bianchi-1} 
\end{align}
where ${G}_{ab}:={R}_{ab}-\frac{1}{2}{R} {g}_{ab}$ is the Einstein tensor of the metric ${g}_{ab}$, and $T_{a b}$ is the stress energy tensor of a Yang--Mills field defined by
\begin{align}\label{e:stregy}
	\mathcal{T}_{a b} =  F_{a}{}^{ c} \cdot  F_{b c} - \frac{1}{4}  g_{a b}   F^{c d}\cdot  F_{c d}.
\end{align}
Here ``$\cdot$'' is the Ad-invariant positive scalar product of $\mathcal{G}$. The stress energy tensor $\mathcal{T}_{a b}$
is divergence free (the conservation law) if the Yang–Mills equations hold (Ref. \cite[\S $6.3$]{Choquet-Bruhat2009} or \cite[\S $3.2$]{Ghanem2014}).

\subsection{Notations and basic concepts}\label{sec-Pre} 	
\subsubsection{Abstract index notations and brackets}\label{s:AIN}
During this paper, we use the \textit{abstract index notations} $a, b, c \cdots$ to index the tensors (e.g., see \cite{Wald2010}).  
We always use $g^{ab}$ and $g_{ab}$ to raise or lower the indices of spacetime tensors (correspondingly, it is equivalent to use the \textit{induced metric} $h_{ab}$ and $h^{ab}$, see \S\ref{s:geohysf} below, to lower or raise for spatial tensors tangent to a spacelike hypersurface $\Sigma$) if there is no explicit indication. In order to express systems in this article in terms of \textit{coordinate-independent formulations}, we introduce an arbitrarily given reference metric $\ulg_{ab}$ on the spacetime $\mathcal{M}^{n+1}$. To simplify the calculations, without loss of generality, we assume $\ulg_{ab}=-(dt)_a(dt)_b+ \ulh_{ab}$ (with zero shift), and $\mathcal{M}^{n+1}=I\times\Sigma$, $I\subset\Rbb$ an open interval, $\Sigma$ is an $n$-dimensional closed Riemannian manifold with metric $\ulh_{ab}$. From now on, we will take this specific reference background throughout this article.  

Denote $\nabla_{c}$, $\tensor{R}{^a_{cde}}$ the covariant derivative and Riemann curvature tensor associating to the metric $g_{a b}$, and $R_{a b}$ the Ricci curvature.   
In addition, we add a bar to denote the ones with respect to the induced metric $h_{ab}$ (defined later) on $\Sigma$, for instance, $\nbar_c, \, \tensor{\bar R}{^a_{cde}}$ are the covariant derivative and Riemann curvature tensor of $h_{a b}$. Correspondingly, underlined notations $\nb_c$, $\underline{R}_{a b}$ refer to the covariant derivative, Ricci curvature associating to the reference metric $\ulg_{a b}$. 
Finally, let $\tensor{\Gamma}{^c_{a b}}$, $\tensor{\underline{\Gamma}}{^c_{a b}}$ be the Christoffel symbols of $g_{a b}$ and $\ulg_{a b}$ respectively.

The bracket on the index indicates the symmetric part of a tensor, for example, 
\begin{align*}
	A_{[a b c]} =   \frac{1}{6}(A_{a b c} + A_{b c a} + A_{c a b}).  
\end{align*}

\subsubsection{Connections}\label{s:conx}
We denote $\tensor{X}{^a_{bc}}$ the $(1,2)$-tensor\footnote{In a local coordinate system, $\tensor{X}{^c_{a b}}=\tensor{\Gamma}{^c_{a b}}-\tensor{\underline{\Gamma}}{^c_{a b}}$.  
} 
satisfying
\begin{equation*}
	\nabla_a  \tensor{\Phi}{^{b_1\cdots b_k}_{c_1\cdots c_l}}=\udn{a}  \tensor{\Phi}{^{b_1\cdots b_k}_{c_1\cdots c_l}}+\sum_i \tensor{X}{^{b_i}_{ad}}\tensor{\Phi}{^{b_1\cdots d \cdots b_k}_{c_1\cdots c_l}}-\sum_j \tensor{X}{^d_{ac_j}}\tensor{\Phi}{^{b_1\cdots b_k}_{c_1\cdots d \cdots c_l}}
\end{equation*}
for any $(k,l)$-tensor $\Phi$ on the spacetime. We define the vector field
\be\label{def-X}
X^a :=  g^{bc} \tensor{X}{^a_{bc}}.
\ee
For later use, direct calculations yield
\al{X2}{
	\tensor{X}{^a_{bc}}
	=  -\frac{1}{2}\bigl(g_{ec}\udn{b}g^{ae}+g_{be}\udn{c}g^{ae}-g^{ae}g_{bd}g_{cf}\udn{e}g^{df}\bigr)  \AND
	X^a =  -\udn{e}g^{ae}+\frac{1}{2} g^{ae}g_{df}\udn{e}g^{df}.
}

\subsubsection{Norms}
We denote by $|\cdot|$ the norm induced by the Ad-invariant scalar product $\cdot$ (recall this product in \eqref{e:stregy}), i.e., $|a|^2:=a\cdot a$ for any $a\in \mathcal{G}$. We also define an inner product $\la v,u \ra_{\ulh}$ with respect to the $\ulh$ whose detailed definitions will be given when we need it (see \eqref{e:inprod} in \S\ref{sec-sym}). The definitions of Sobolev spaces in this article comes from Choquet-Bruhat \cite[Appendix I]{Choquet-Bruhat2009} with the Riemmanian manifold being replaced by $(\Sigma,\ulh)$ and the covariant derivatives being the corresponding one. We remark that if the variable $a$ is a $\mathcal{G}$-valued tensor, then its Sobolev norms are defined by $\||a|\|_{W^{s,p}}$, and for simplicity, we still denote it as $\|a\|_{W^{s,p}}$, the Ad-invariant norm will be clear from the contexts.	

\subsubsection{Geometry of hypersurfaces}\label{s:geohysf}
Let $t$ be a global time function on the spacetime $\mathcal{M}^{n+1} = I \times \Sigma$.
The gradient $\nabla^a t=g^{ab}\nabla_b t=g^{ab}(dt)_b$ is normal to the slice $\Sigma_t:=\{t\} \times \Sigma$  with respect to $g_{a b}$. Then 
\be\label{def-T}
T^a := (-\lambda)^{-\frac{1}{2}} \nabla^a t, \quad \text{with} \quad \lambda = g_{ab} \nabla^a t \nabla^b t= g^{ab} \nabla_a t \nabla_b t
\ee is the unit timelike vector field normal to $\Sigma_t$. 	 
The induced spatial metric on $\Sigma_t$ is given by
\be\label{def-h}
h_{a b} := g_{a b} + T_a T_b, 
\ee
and $h^{a b} := g^{a b} + T^a T^b$ is the inverse, where $T_a:=T^b g_{ab}=(-\lambda)^{-\frac{1}{2}} \nabla_a t$. In addition,  we define the projection onto $\Sigma_t$, 
\begin{align}\label{e:def-h2}
	\tensor{h}{^c_d} := \tensor{\delta}{^{c}_{ d}} + T^c T_d.
\end{align}

Let us define the \textit{second fundamental form} \[k_{ab} := \frac{1}{2}\lie_T h_{ab}.\] It holds that 
\begin{align}\label{e:2ndfm}
	k_{ab}=\nabla_b T_a+T_b \nabla_T T_a = \nbar_a T_b = \nbar_b T_a,
\end{align}
where we use the notation $\nabla_T:=T^c\nabla_c$. We also denote $\tr k := h^{a b} k_{a b}$ the trace of $k_{ab}$. 

Let $\Psi$ be any $(r,s)$ tensor on the manifold $\mathcal{M}^{n+1}$. We call $\Psi$ \textit{$\Sigma$-tangent} (or \textit{spatial tensor}) if \[\tensor{\Psi}{^{a_1 \cdots a \cdots a_r}_{b_1\cdots b_i \cdots b_s}} T_a = 0 \quad \text{and} \quad \tensor{\Psi}{^{a_1 \cdots a_i \cdots a_r}_{b_1\cdots b \cdots b_s}} T^b = 0.\] As a remark, if a $(0, r)$ tensor $\Psi$ is $\Sigma$-tangent, then the Lie derivative $\lie_T \Psi$ is $\Sigma$-tangent as well (see \eqref{e:lienb} later).

In addition, in this article, we also need the normal vector $\nu_c:=\nb_c t=(dt)_c$ and $\nu^a:=\ulg^{ab}\nu_b=-(\partial/\partial t)^a$ with respect to the reference metric $\ulg_{ab}$, and we define the background projection by
\begin{equation}\label{e:ulhproj}
	\tensor{\ulh}{^c_d} := \tensor{\delta}{^{c}_{ d}} + \nu^c \nu_d. 
\end{equation}

\subsubsection{Variables}
In the new symmetric hyperbolic formulation, we employ the following variables 
\be\label{vari-metric}
g^{a b}, \quad \tensor{g}{^{a b}_d} := \nb_d g^{a b}, 
\ee
and
\be\label{vari-YM}
A_a, \quad E_b:= \tensor{h}{^a_b}  F_{a p} T^p, \quad H_{db}:=\tensor{h}{^c_d} F_{c a} \tensor{h}{^a_b}, \AND \tE^a := E_b g^{b c} \tensor{h}{^a_c}. 
\ee
Denote the compacting variable 
\be\label{def-U}
\mathbf{U}: 
= \left(- \nu^d \tensor{g}{^{a b}_d}, \,
\tensor{\ulh}{^d_{c}} \tensor{g}{^{a b}_d}, \, g^{ab}, \, \tE^a, \, E_a, \, H_{db}, \,  A_d \right)^T.
\ee

\subsubsection{Temporal gauges}\label{s:temgg}
The \textit{temporal gauge} is described as
\begin{equation}\label{e:temgg}
	A_a T^a = 0.
\end{equation} 
Note that, in the temporal gauge, $A_a =\bar{A}_a := A_b \tensor{h}{^b_a}$.

\subsection{Main result}\label{s:mainthm}
In this section, let us state the main theorem and its proof is postponed to \S\ref{s:mainpf}.

\begin{theorem}\label{t:mainthm}
	Let $(\Sigma,\, g_0,\, k_0, \, A_0, \, E_0)$  
	be an initial data set satisfying the Einstein and Yang--Mills constraints. Suppose $(g_0, k_0 ) \in H^{s+1} (\Sigma) \times H^{s} (\Sigma)$ and $( A_0, E_0) \in H^{s} (\Sigma) \times H^{s} (\Sigma)$ for $s\in\Zbb_{>\frac{n}{2}+1}$. Then 
	
	$(a)$ (Local existence) There is a constant $t_\ast >t_0$ and a  unique maximal development $(\mathcal{M}^{n+1}= I \times \Sigma, \, g, \, A, \, E)$, $I \subset \mathbb{R}$, for the Einstein--Yang--Mills system in the temporal gauge, such that 
	\begin{equation*}
		g \in \bigcap_{\ell=0}^{s}C^\ell([t_0,t_\ast),H^{s-\ell+1}(\Sigma ))\AND E, A\in \bigcap_{\ell=0}^{s }C^\ell([t_0,t_\ast),H^{s-\ell}(\Sigma )),
	\end{equation*}
where $E, \, H$ are defined by \eqref{vari-YM}. 	

	$(b)$ (Continuation principle) If
	\begin{equation*}
		\| \left( \nu^d \tensor{g}{^{ab}_d}, \, \tensor{\ulh}{^d_{c}} \tensor{g}{^{a b}_d}, \, g^{ab}, \, E_a, \, A_d, \, H_{a b} \right)\|_{\Li ([t_0,t_\star), W^{1,\infty}(\Sigma))}<\infty,
	\end{equation*}
	then the solution $\left( \nu^d \tensor{g}{^{ab}_d}, \, \tensor{\ulh}{^d_{c}} \tensor{g}{^{a b}_d}, \, g^{ab}, \, E_a, \, A_d \right)$ can be uniquely continued, as a classical solution with the same regularity, to a larger time interval $t\in[t_0, t^\star)$ where $t^\star\in(t_\star, + \infty)$.
\end{theorem}

\subsection{Related works}
Since the EYM system is a coupled system of the Einstein and Yang--Mills equations, let us briefly review relevant results respectively.
The local well-posedness for the Cauchy problem of vacuum Einstein equations was initiated by Choquet-Bruhat \cite{Choquetbruhat2021} via introducing the wave gauge and solving the reduced Einstein equations. Based on the wave gauge or variants of the generalized wave gauge, there were ample results concerning the long time stability of gravity along the flow of vacuum Einstein equations or Einstein equations coupled with various matter fields, for instance,   \cite{Ringstroem2008, Lindblad2005b, Lindblad2010,Liu2018,Liu2018b,Liu2018a,Wang2019}. A thorough review for the Cauchy problem of  the Einstein equations can be found in  \cite{Choquet-Bruhat2009}.

In regard to the Yang--Mills equations, the local well-posedness with different gauges, like Lorentz gauge, temporal gauge, had been known \cite{Kerner1974,Ginibre1982,Segal1979,ChoquetBruhat1991,Eardley1982,Ginibre1981,ChoquetBruhat1982a}. Among these results, we focus on the ones with temporal gauge. Using the nonlinear semi group theory, Segal \cite{Segal1979} proved the local well-posedness of the Yang--Mills equations with temporal gauge in the $4$-dimensional Minkowski spacetime. It was later advanced by Eardley-Moncrief \cite{Eardley1982} and was extended by Ginibre-Velo \cite{Ginibre1981} to the $(n+1)$-dimensional Minkowski spacetime. A more general proof that works on a non-flat manifold was carried out by Choquet-Bruhat and Segal \cite{ChoquetBruhat1982a} where a third order Leray hyperbolic system was derived  for the Yang--Mills connection in temporal gauge.

In this paper, we revisit the Cauchy problem for the EYM equations, but within the temporal gauge for the Yang--Mills field.  In the proof, a first order symmetric hyperbolic formulation for the $(n+1)$-dimensional Yang--Mills equations in the temporal gauge is essentially new, which, together with the Einstein equations in the wave gauge, establishes the local well-posedness of the $(n+1)$-dimensional EYM system in the wave gauge (for metric) and the temporal gauge (for Yang--Mills field).

\subsection{An overview of the proof}
To demonstrate a symmetric hyperbolic formulation for the $(n+1)$-dimensional EYM system with the temporal gauge, a crucial step is symmetrizing the Yang--Mills system. An innovation in this paper lies in that, instead of barely appealing to equations for the connection (or coupled with the ``electric'' Yang--Mills field equation), we additionally take the full ``electric'' and ``magnetic'' fields equations into accounts. Using this idea, the natural hyperbolic formulation for the $4$-dimensional Maxwell field equations \cite{Racke2015} will suggest a straightforward extension to the $4$-dimensional Yang--Mills case with temporal gauge. Moreover, adapted to the arbitrarily dimensional Yang--Mills case, we propose a new hyperbolic formulation that combines the field equations \eqref{eq-YM-bianchi-0}, \eqref{eq-YM-div-0} with the dynamic equation for the connection. It is crucial that this new hyperbolic formulation fits well into the long time scheme of our companion work \cite{LW2021a}  with Todd A. Oliynyk.
Apart from the difficulty of symmetrization, the ``electric'' and ``magnetic'' fields are strongly coupled in the Yang--Mills field equations (see Lemma \ref{lem-maxwell-hyperbolic-0}), which we eventually  formulate as a system over tensor bundles. Thus we have to extend the local well-posedness result of a hyperbolic system to the one over tensor bundles. Since the hyperbolic symmetrization for the Einstein equations follows from the standard procedure, we only focus on the Yang--Mills part in this overview.  Let us sketch the main idea of symmetrizing the $(n+1)$-dimensional Yang--Mills equations in the temporal gauge.

\underline{Step $1$:} 
The most important preparation for this procedure is the introduction of the symmetrizing tensors\footnote{ Private communications with Todd A. Oliynyk and introduced in our companion article \cite{LW2021a}.  } \[\underline{Q}^{edc}:= \nu^e g^{dc} + \nu^d g^{ec} -\nu^c g^{ed} \AND Q^{edc} := T^e g^{dc} + T^d g^{e c} - T^c g^{e d}, \] which turns out to be effective for the Einstein and Yang--Mills equations both. In order to symmetrize the Yang--Mills equations, we use the $Q^{edc}$ tensor to unify the Yang--Mills propagation equations (i.e., in \eqref{eq-YM-div-0} and \eqref{eq-YM-bianchi-1}) to a single first order system (see Lemma \ref{lem-maxwell-hyperbolic-0}), which combining with the dynamic equation for the connection $A$ constitutes a closed first order system. Note that at this step, we have included a new variable, the ``magnetic'' component of the Yang--Mills field $H_{db}:=\tensor{h}{^c_d} F_{c a} \tensor{h}{^a_b}$,  
and the dynamic Yang--Mills Bianchi identities in the first order system that has been stated in Lemma \ref{lem-maxwell-hyperbolic-0}.

\underline{Step $2$:} Unfortunately, this first order system fails to be symmetric when $n+1>4$ (in the $n+1=4$ case,  the Hodge dual variant of the ``magnetic'' field trivially leads to a symmetric system, which is in analogy with the $4$-dimensional Maxwell case in \cite{Racke2015}). To overcome this difficulty, we introduce a new variable $\tE^a := E_b g^{b c} \tensor{h}{^a_c}$ (recalling \eqref{vari-YM}) as one of the unknowns and add a dynamic equation of $\tE^a$ which is derived from the equation of $E_b$ (essentially no new equation is introduced). In fact, the dynamic equation of $\tE^a$ (derived from the above first order system) serves to compensate the non-symmetric part of the system in Step 1. The detailed information for this step is included in Lemma \ref{lem-FOSHS-YM}.

\underline{Step $3$:} After the preceding two steps, we achieve the target first order symmetric hyperbolic system \eqref{YM-FOSHS-1} on the variables $(A_a, \, E_b, \, \tE^e, \, H_{c d})$.  Then a key step next is to build the equivalence between the symmetric hyperbolic system and the EYM equations (enclosed in Theorem \ref{t:depctrnt}). To achieve this, a main aim is to conclude: 
\begin{claim}
	The solution of the target first order symmetric hyperbolic system \eqref{YM-FOSHS-1} uniquely determines a solution of the Yang--Mills system if the initial data are induced by Yang--Mills fields.
\end{claim}	
In fact, from the definitions of $\tE^a$ and $H_{ab}$ (recalling \eqref{def-F} and \eqref{vari-YM}), we have shown that if $(A_a, \, E_b)$ is a solution of the Yang--Mills equations, then $(A_a, \, E_b, \, \tE^e, \, H_{c d})$ with the extra variables $\tE^e$ and $H_{c d}$ \footnote{They are defined in terms of $E_b$ and $A_a$ and we denote these transformations by $(E_b\leftrightarrow \tE^a)$ and $(A_b\leftrightarrow H_{ab})$ for short, respectively, see \eqref{def-F} and \eqref{vari-YM}. } satisfies the target first order symmetric hyperbolic system \eqref{YM-FOSHS-1} since the extra equations for $\tE^e$ and $H_{c d}$ are inferred from the Yang--Mills equations and the transformations of $(E_b\leftrightarrow \tE^a)$ and $(A_b\leftrightarrow H_{ab})$. This result tells that \textit{if there is} a solution to the Yang--Mills equations, then it must be unique since this solution solves the target system \eqref{YM-FOSHS-1} which has a unique solution due to the symmetric hyperbolicity. It therefore concludes the uniqueness of the Yang--Mills solution.

The rest of task is to confirm the Yang--Mills solution does \textit{exist}. For this purpose, we will use the solution $(A_a, \, E_b, \, \tE^e, \, H_{c d})$ of the target system\footnote{Note that the relations between $A_a$ and $H_{a b}$, $E_b$ and $\tE^a$ are a priori unknown in the hyperbolic system, although they will be recovered with more effort on their dynamics. 
} \eqref{YM-FOSHS-1} (whose existence and uniqueness are implied by the symmetric hyperbolicity)  to construct a solution of the Yang--Mills system. In doing this, we have to impose the Yang--Mills constraints on $A_a$ and $E_b$ to the data of the target system \eqref{YM-FOSHS-1}. Additionally, we complement the initial constraints $(E_b\leftrightarrow \tE^a)|_{t=t_0}$ and $(A_b\leftrightarrow H_{ab})|_{t=t_0}$ to the data on $\Sigma_{t_0}$. Then the developed lemma on these constraints (see Lemma \ref{t:ctrtevl}) confirms that the solution $(A_a, \, E_b, \, \tE^e, \, H_{c d})$ to the target system \eqref{YM-FOSHS-1} with the specific data makes sure the constraints are maintained during the propagation of the hyperbolic system. That is, in terms of Lemma \ref{t:ctrtevl}, the following constraints
\begin{align*}
	\alpha^d :={}& \tE^e \tensor{h}{^d_e} + h^{c d} E_{c}\equiv0, \\
	\beta_{p q}: ={}& \nbar_p A_{q} - \nbar_q A_p + [A_p, A_q] - H_{p q}\equiv0, \\
	\gamma :={}& \nbar^a E_a + h^{a b} [A_a, E_b]\equiv0
\end{align*}
hold true in the existence interval $t\in[t_0,t_\ast)$, and hence they reduce the target hyperbolic system to the Yang--Mills system.  
In other words, the target hyperbolic system \eqref{YM-FOSHS-1} with specific initial constraints being satisfied recovers the Yang--Mills system.

Eventually, we further analyze the local existence and uniqueness of the geometric solution of the hyperbolic system over tensor bundles, see \S\ref{s:local}. Moreover, using this model and  the equivalence Theorem \ref{t:depctrnt}, we concludes the solution of the local existence and uniqueness of solutions to the EYM equations.

\subsection{Outlines}
We arrange this paper as follows. 
The new symmetric hyperbolic formulation of the EYM system in the wave gauge for the metric and temporal gauge for the Yang--Mills field is carried out in \S \ref{s:hyper-EYM}. In \S \ref{s:local}, we prove the local well-posedness for a model of hyperbolic system over tensor bundles. In the end, we confirm  in \S\ref{s:mainpf} the local existence theorem of the EYM system. More useful calculations are collected in the appendix.


\section{A new symmetric hyperbolic formulation for the EYM system}\label{s:hyper-EYM}

\subsection{Symmetric hyperbolic formulations for the reduced Einstein equations} 
In this section, we recall the hyperbolic formulation for the Einstein equations based on wave gauge.
Before going into the EYM system, let us introduce the symmetrizing tensors, a useful tool for symmetrizing the EYM system, defined by
\be\label{def-Q}
\underline{Q}^{edc}:= \nu^e g^{dc} + \nu^d g^{ec} -\nu^c g^{ed} \AND Q^{edc} := T^e g^{dc} + T^d g^{e c} - T^c g^{e d}.
\ee  
The $Q^{e d c}$ are symmetric in the first two indices $e$ and $d$, i.e. $Q^{edc} = Q^{dec}$, and more of the computing properties are presented in Proposition \ref{lem-identity} (see Appendix \ref{s:App1}). Similar conclusions hold for  the one defined in terms of $\nu$, $\underline{Q}^{edc}$.

Since \cite{Choquetbruhat2021} when Choquet-Bruhat introduced the \textit{wave gauge}, it becomes a well known technique in solving the Einstein equations. Then the wave gauge in this article is chosen by requiring (recalling $X^a$ in \eqref{E:X2}) 
\al{CONSTR1}{X^a=f^a}
where $f^a$ is a given function.   
A standard argument shows that to solve the Einstein equations, it suffices to solve the reduced Einstein system with wave gauge and further the reduced system can be made into a hyperbolic system. We briefly include this standard procedure on the Einstein equations in the following Lemma \ref{t:rdein} and several variants can be found in many references (see, for example, \cite[Chapter VI, \S $7.4$]{Choquet-Bruhat2009} and \cite{LW2021a}). 

Let us recall that we have introduced an arbitrarily given reference metric $\ulg_{ab}$ and the compatible covariant derivative is denoted by $\nb_c$(see \S\ref{s:AIN}), which will be helpful in writing down a coordinate-independent formulation.  
We list the results in what follows without proof, and for the readers' convenience, brief proofs are provided in Appendix \ref{s:App1}. Firstly, using the wave gauge and the expansions of the Ricci tensor in terms of the reference metric, we arrive at the following lemma. 
\begin{lemma}\label{t:rdein}
	Suppose $\udn{c}$ is any covariant derivative compatible to the reference metric $\ulg_{ab}$, under the wave gauge \eqref{E:CONSTR1}, the reduced Einstein equations with the Yang--Mills source become
	\be\label{eq-redu-Einstein}
	g^{cd}\udn{c}\udn{d} g^{ab} = S^{a b}(\mathbf{U}), 
	\ee
	where
	\als{
		S^{a b}(\mathbf{U}) = {} & -2\nabla^{(a} f^{b)} -2 g^{bd} E^a E_d + 2 g^{bd}H^{a b} H_{d b} +  2 g^{bd}T_d H^{a b} E_b  + 2 g^{bd} g^{ac}T_c T_d |E|^2  \nnb \\
		&-\frac{1}{(n-1)} g^{ab} (-2 |E|^2 + |H|^2) -2\ulR^{ab}-2P^{ab}-2Q^{ab} +\frac{4\Lambda}{n-1}  g^{ab},    
	} 
	and $P^{a b}$ depending on $g^{-1} - \ulg^{-1}$ (up to quadratic), $Q^{ab}$ being quadratic in $\nb g^{-1}$ are described in Lemma \ref{t:Rdop}.
\end{lemma}

\begin{theorem}\label{thm-hyperb-gravity}
	The reduced Einstein equations \eqref{eq-redu-Einstein} can be rewritten as a hyperbolic system
	\al{Einstein-sym-thm}{
		-\bar{\mathbf{A}}^0 \nu^c \nb_c \p{  \tensor{g}{^{a b}_d}  \\ g^{a b} }+\bar{\mathbf{A}}^c \tensor{\ulh}{^	b_c} \nb_b \p{  \tensor{g}{^{a b}_d}  \\ g^{a b} } =  \bar G (\mathbf{U}),  
	} 
	where, (recalling \eqref{def-Q} for  $\underline{Q}^{edc}$), 
	\als{
		\bar{\mathbf{A}}^0 =
		\p{ \underline{Q}^{ed\hc}\nu_{\hc}  & 0 \\ 0 & 1 }
		\AND \bar{\mathbf{A}}^c \tensor{\ulh}{^	b_c} =
		\p{ \underline{Q}^{edc} \tensor{\ulh}{^b_c}  & 0 \\ 0 & 0 }
	}
	and  
	\[\bar G (\mathbf{U}) = \p{ \nu^e S^{a b}(\mathbf{U}) +\nu^d g^{ec}(\tensor{\underline{R}}{_{cdf}^a} g^{fb}+\tensor{\underline{R}}{_{cdf}^b} g^{fa})\\ \tensor{g}{^{a b}_d}\nu^d }. \]
\end{theorem}

In order to fit into the framework of Theorem \ref{t:lcexuqct} conveniently, we transform the above reduced Einstein equations \eqref{E:Einstein-sym-thm} into the following form.
\begin{corollary}\label{t:Ein2}
	The reduced Einstein equations \eqref{E:Einstein-sym-thm} can be further rewritten as a hyperbolic system
	\al{Einstein-sym-thm2}{
		-\check{\mathbf{A}}^0 \nu^c \nb_c \p{  -\nu^d \tensor{g}{^{a b}_d}  \\
			\tensor{\ulh}{^d_{\hd}} \tensor{g}{^{a b}_d}  \\ g^{a b} }+\check{\mathbf{A}}^c \tensor{\ulh}{^b_c} \nb_b \p{  -\nu^d \tensor{g}{^{a b}_d}  \\
			\tensor{\ulh}{^d_{\hd}}\tensor{g}{^{a b}_d}  \\ g^{a b} } =  \check G (\mathbf{U}),  
	} 
	where
	\als{
		\check{\mathbf{A}}^0 =
		\p{ -\nu_eg^{ec} \nu_c & 0  & 0 \\ 
			0 & \tensor{\ulh}{_{e \hc}} g^{ed} \tensor{\ulh}{^{\he}_d} & 0 	\\
			0 & 0 & 1 }
		\AND \check{\mathbf{A}}^c \tensor{\ulh}{^b_c} =
		\p{ -\tensor{\ulh}{^b_c} g^{dc}\nu_d  & -g^{dc} \tensor{\ulh}{^{\he}_d} \tensor{\ulh}{^b_c} & 0  \\ -\ulh_{\hc e} g^{ec} \tensor{\ulh}{^b_c} & 0 & 0 \\
			0 & 0 & 0}
	}
	and  
	\[\check G (\mathbf{U}) =  \p{\nu_e & 0 \\ \ulh_{e \hc} & 0 \\ 0 & 1 }\bar{G}(\mathbf{U}). \]
\end{corollary} 
\begin{proof}
	Direct calculations imply \al{NBNU}{
		\udn{a}\nu_b=0, \quad \udn{a}\nu^b=0  \AND \udn{c}\tensor{\ulh}{^a_b}=0,
	}
	and the decomposition
	\begin{align}\label{e:gdcp}
		\p{ \tensor{g}{^{ab}_d} \\ g^{ab}}=\p{\nu_d & \tensor{\ulh}{^{\he}_d} & 0 \\ 0 & 0 & 1}\p{-\nu^e \tensor{g}{^{ab}_e} \\ \tensor{\ulh}{^e_{\he}} \tensor{g}{^{ab}_e} \\ g^{ab}}. 
	\end{align}
	Using 
	\begin{align*}
		\p{1 & 0 & 0 \\ 0 & \ulh_{f\hc} & 0 \\
			0 & 0 & 1 }\p{\nu_e & 0 \\ \ts{\ulh}{^f_e} & 0 \\ 0 & 1 }
	\end{align*} 
	to act on \eqref{E:Einstein-sym-thm} and by \eqref{E:NBNU}-\eqref{e:gdcp}, we prove this corollary. 
\end{proof}

\subsection{Symmetric hyperbolic formulations of Yang--Mills equations}\label{sec-Max}
In this section, we conduct an ``extension'' for the Yang--Mills equations. Specifically, we derive a symmetric hyperbolic system which  reduces to the Yang--Mills equations if the data set is induced from a Yang–Mills field. Recall that the Yang–Mills curvature (see \eqref{eq-YM-div-0}--\eqref{eq-YM-bianchi-1}) satisfies
\begin{align}
	\nabla^a F_{a b} &= -  g^{a c} [A_c, F_{a b} ], \label{eq-YM-div} \\
	\nabla_{[a} F_{b c]} &= - [A_{[a}, F_{b c]} ]. \label{eq-YM-Bianchi}
\end{align}
Then, \eqref{eq-YM-div} and  \eqref{eq-YM-Bianchi} consist of \textit{propagation equations}
\begin{align}
	g^{b a} \nabla_b F_{a c} \tensor{h}{^{c}_{ e}} & = - g^{a d} \tensor{h}{^{c}_{ e}} [A_d, F_{a c} ] , \label{YM-div-dyn} \\
	\nabla_{[b} F_{a c ]} T^b \tensor{h}{^a_{p}} \tensor{h}{^c_q} &  = - T^b \tensor{h}{^a_{p}} \tensor{h}{^c_q} [A_{[b}, F_{a c]} ], \label{YM-bianchi-dyn}
\end{align}
and \textit{constraint equations} 
\begin{align}
	g^{b a} \nabla_b F_{a c} T^{c} & =  - g^{a d} T^{c} [A_d, F_{a c} ], \label{YM-div-F-const} \\
	\nabla_{[b} F_{a c ]} \tensor{h}{^b_d} \tensor{h}{^a_{p}} \tensor{h}{^c_q} &  = -  \tensor{h}{^b_d} \tensor{h}{^a_{p}} \tensor{h}{^c_q} [A_{[b}, F_{a c]} ]. \label{YM-bianchi-F-const}
\end{align}

We work within the \textit{temporal gauge}:
\be\label{temporal}
A_a T^a = 0,
\ee
i.e. \[A_a g^{a b} T_b = 0.\]

In the next two lemmas, we take advantage of the symmetric tensor $Q^{edc} := T^e g^{dc} + T^d g^{e c} - T^c g^{e d}$ to unify the Yang--Mills propagation equations \eqref{YM-div-dyn}-\eqref{YM-bianchi-dyn} and the dynamic equation for $A$ into a hyperbolic system. 
Recall the decomposition in \eqref{vari-YM}: \[E_b:= - T^p F_{p a} \tensor{h}{^a_b}, \quad \text{and} \quad H_{db}:=\tensor{h}{^c_d} F_{c a} \tensor{h}{^a_b}. \] As the first step, we deduce a complete first order system for $E_a$, $H_{ab}$ and $A_a$.

\begin{lemma}\label{lem-maxwell-hyperbolic-0}
	The Yang--Mills system \eqref{YM-div-dyn}--\eqref{YM-bianchi-dyn} for any metric $g_{ab}$ implies the following propagation equations
	\begin{align}\label{e:maineq1}
		\A^{c}  \nabla_c
		\begin{pmatrix}
			E_{f} \\  H_{e f} \\  A_{f}
		\end{pmatrix}
		= {}&
		\begin{pmatrix}  \widetilde{\Delta}_{1 a}  \\ \widetilde{\Delta}_{2 b}^{h} \\ \widetilde{\Delta}_{3 \hc} \end{pmatrix},
	\end{align}
	where
	\begin{align*}
		\A^c & =
		\begin{pmatrix}
			-T^c \tensor{h}{^f_{a}} & - h^{e c} \tensor{h}{^f_{a}}  & 0 \\
			\tensor{h}{^{c}_{b}} h^{h f} - h^{h c} \tensor{h}{^f_{b}}  & -T^c h^{h e} \tensor{h}{^f_{b}} & 0 \\
			0 & 0 & -T^c \tensor{h}{^f_{\hc}}
		\end{pmatrix}
	\end{align*}		
	and $\widetilde{\Delta}_i$, $i=1,2,3$, are defined as below
	\begin{gather*}
		\widetilde{\Delta}_{1 a} =  h^{c d} [A_d, H_{c a}] + h^{c d} \nabla_c T_d E_{a} - \tensor{h}{^f_{a}} h^{d c} \nabla_c T_f E_d + \tensor{h}{^f_{a}} \nabla_T T^e H_{e f}, \\
		\widetilde{\Delta}_{2 b}^{h} =  h^{h d} \tensor{h}{^a_{b}} ([A_{d}, E_{a} ] - [A_a, E_d])  - h^{h d}  \tensor{h}{^a_{b}} (\nabla_T T_a E_{d} -  \nabla_T T_{d} E_{a}) + ( \tensor{h}{^f_{b}} h^{h c} - \tensor{h}{^c_{b}} h^{h f}  ) \nabla_c T^e H_{e f}, \\
		\widetilde{\Delta}_{3 \hc}  =   E_{\hc} + \tensor{h}{^a_{\hc}} \nabla_a T^{d}  A_{d}. 
	\end{gather*}	
\end{lemma}
\begin{remark}\label{rk-max-nonsym}
	Although the coefficient matrix 
	\begin{align*}
		\A^c T_c & =
		\begin{pmatrix}
			\tensor{h}{^f_b} & 0  & 0 \\
			0  &  h^{h l} & 0 \\
			0 & 0 &  \tensor{h}{^f_{\hc}} 
		\end{pmatrix}
	\end{align*}
	is positive definite and symmetric, unfortunately, $\A^c \tensor{h}{^b_c}$ is non-symmetric. Therefore, \eqref{e:maineq1} fails to constitute a hyperbolic system. Observe that the term $\tensor{h}{^{c}_{b}} h^{h f}$ in the first column and the second row of $\A^c$  breaks the symmetry and this term comes from $\tensor{h}{^{c}_{b}} h^{h f} \nabla_{c} E_f$ in the system \eqref{e:maineq1}. It motivates us to introduce a new variable $\tE^e:=-h^{e f} E_f$ in order to symmetrize this system in the next step.
\end{remark} 
\begin{proof}
	The Yang--Mills equation \eqref{YM-div-dyn} gives rise to
	\begin{equation*}
		Q^{edc} \nabla_c F_{d a} \tensor{h}{^{a}_{b}} + (- T^d g^{e c} + T^c g^{e d} ) \nabla_c F_{d a} \tensor{h}{^{a}_{b}}  =  - T^e g^{c d} [A_d, F_{c a} ] \tensor{h}{^{a}_{b}}.
	\end{equation*}
	Making use of the Yang--Mills Bianchi equation \eqref{YM-bianchi-dyn} and 	
	due to the anti-symmetry of the Yang--Mills field $F_{ab}$, yields
	\begin{align*}
		(- T^d g^{e c} + T^c g^{e d} ) \nabla_c F_{d a}  \tensor{h}{^{a}_{b}}  
		= &  T^d g^{e c} (\nabla_c F_{a d} + \nabla_d F_{c a} )  \tensor{h}{^{a}_{b}} \\
		= & - T^d g^{e c} \nabla_a F_{d c}  \tensor{h}{^{a}_{b}} - T^d g^{e c} \tensor{h}{^{a}_{b}} [A_{[d}, F_{c a]} ]  \\
		= & - T^d g^{e q} \tensor{h}{^f_q} \tensor{h}{^a_f} \nabla_c F_{d a} \tensor{h}{^{c}_{b}}  - T^d g^{e c} \tensor{h}{^{a}_{b}} [A_{[d}, F_{c a]} ],
	\end{align*}
	where in the last equality, we note that, 
	\begin{align*}	
		- T^d g^{e a} \nabla_c F_{d a} \tensor{h}{^{c}_{b}} = {}&  - T^d g^{e f} \tensor{h}{^a_f} \nabla_c F_{d a} h^c_b =  - T^d g^{e q} \tensor{h}{^f_q} \tensor{h}{^{a}_{f}} \nabla_c F_{d a} \tensor{h}{^{c}_{b}}. 
	\end{align*}
	As a result, we arrive at a first order system 
	\begin{align}
		& \tensor{h}{^{f}_{b}} Q^{edc} \nabla_c (F_{d a} \tensor{h}{^{a}_{f}})  - \tensor{h}{^{c}_{b}} T^d g^{e q} \tensor{h}{^f_q} \nabla_c (F_{d a} \tensor{h}{^{a}_{f}}) = \widetilde{\Delta}^{\prime e}_b, \label{YM-0}
	\end{align}
	where 
	\begin{align*}
		\widetilde{\Delta}_b^{\prime e} ={} & \tensor{h}{^{a}_{b}} \left( T^d g^{e c}  [A_{[d}, F_{c a]} ] -  T^e g^{c d}  [A_d, F_{c a} ] \right) + ( \tensor{h}{^{f}_{b}} Q^{edc} - \tensor{h}{^{c}_{b}} T^d g^{e q} \tensor{h}{^f_q} ) \nabla_c \tensor{h}{^{a}_{f}}  F_{d a}.
	\end{align*}
	We emphasis that \eqref{YM-0} has taken the full Yang--Mills propagation equations \eqref{YM-div-dyn}--\eqref{YM-bianchi-dyn} into accounts.
	
	Noting the decomposition
	\begin{align*}
		F_{d a} \tensor{h}{^a_f}  
		={}&
		\begin{pmatrix}
			T_d, \tensor{h}{^{\hd}_d}
		\end{pmatrix}
		\begin{pmatrix}
			-T^p F_{p a} \tensor{h}{^a_f} \\ \tensor{h}{^{\ha}_{\hd}} F_{\ha a} \tensor{h}{^a_f}
		\end{pmatrix}
		=
		\begin{pmatrix}
			T_d, \tensor{h}{^{\hd}_d}
		\end{pmatrix}
		\begin{pmatrix}
			E_f \\ H_{\hd f} 
		\end{pmatrix},
	\end{align*}
	and acting $\begin{pmatrix} T_e \\ \tensor{h}{^h_e} \end{pmatrix}$ on \eqref{YM-0}, we then  turn the Yang--Mills equations into the following form
	\begin{align}\label{eq-YM-1}
		&\quad
		\begin{pmatrix}
			\tensor{h}{^{f}_{b}} T_e Q^{edc} T_d  & \tensor{h}{^{f}_{b}} T_e Q^{edc} \tensor{h}{^{\hd}_d}  \\
			\tensor{h}{^{f}_{b}} \tensor{h}{^h_e} Q^{edc} T_d +  \tensor{h}{^{c}_{b}} h^{h f} & \tensor{h}{^{f}_{b}} \tensor{h}{^h_e} Q^{edc} \tensor{h}{^{\hd}_d} 
		\end{pmatrix}
		\nabla_c
		\begin{pmatrix}
			E_f \\ H_{\hd f} 
		\end{pmatrix} 
		= 
		\begin{pmatrix} \widetilde{\Delta}_{1 b} \\  \widetilde{\Delta}^{h}_{2 b}
		\end{pmatrix},
	\end{align}
	where
	\als{
		\widetilde{\Delta}_{1 b} ={}&  T_e \widetilde{\Delta}^{\prime e}_b - T_e ( \tensor{h}{^{f}_{b}} Q^{edc} - \tensor{h}{^{c}_{b}} T^d g^{e q} \tensor{h}{^f_q} ) (\nabla_c T_d E_f + \nabla_c \tensor{h}{^{\hd}_d} H_{\hd f}), \\
		={}& \tensor{h}{^{a}_{b}} \left( T^d T^{ c}  [A_{[d}, F_{c a]} ] + g^{c d}  [A_d, F_{c a} ] \right) + h^{c d} \nabla_c T_d E_b - \tensor{h}{^f_b} h^{d c} \nabla_c T_f E_d + \tensor{h}{^f_b} \nabla_T T^e H_{e f}, \\
		\widetilde{\Delta}^h_{2 b} ={}&  \tensor{h}{^h_e} \widetilde{\Delta}^{\prime e}_b - \tensor{h}{^h_e} (\tensor{h}{^{f}_{b}} Q^{edc} - \tensor{h}{^{c}_{b}} T^d g^{e q} \tensor{h}{^f_q} ) (\nabla_c T_d E_f + \nabla_c \tensor{h}{^{\hd}_d} H_{\hd f})\\
		={}& \tensor{h}{^h_e} \tensor{h}{^{a}_{b}} T^d g^{e c}  [A_{[d}, F_{c a]} ] - \tensor{h}{^f_b} \nabla_T T_f \tensor{h}{^{h \hd}} E_{\hd} + h^{h \hd} \nabla_T T_{\hd} E_b  + ( \tensor{h}{^f_b} h^{h c} - \tensor{h}{^{c}_{b}} h^{h f}  ) \nabla_c T^e H_{e f}. 	 
	}
	
	On the other hand, we substitute the temporal gauge \eqref{temporal} into the definition of the Yang--Mills curvature to derive a dynamic equation for the connection $A_a$,
	\als{
		\nabla_{T} A_a = {}&T^b (\nabla_a A_b -  [A_b, A_a] + F_{b a})  
		={} -E_a - \nabla_a T^{d}  A_{d}.
	} 
	That is, after taking the $\Sigma$-tangent part,
	\begin{equation}\label{eq-YM-potential}
		- \tensor{h}{^{a}_{b}} T^c \nabla_{c} A_a  = E_b + \tensor{h}{^{a}_{b}} \nabla_a T^{d}  A_{d}.
	\end{equation}			
	Putting \eqref{eq-YM-1} and \eqref{eq-YM-potential} together, we complete the proof.
\end{proof}

In correspondence with Remark \ref{rk-max-nonsym}, we  introduce an extra unknown variable
\be\label{def-tE}
\tE^a:=-g^{e b}\tensor{h}{^a_e}E_{b},
\ee
so that the final symmetric hyperbolic system for the Yang--Mills field concerns the unknowns $\tE^e$, $E_{d}$, $H_{a b}$ and $A_b$. 

Now we \textit{aim} to recast the Yang--Mills system in Lemma \ref{lem-maxwell-hyperbolic-0} into a symmetric hyperbolic system with regard to the variables $\tE^e$, $E_{d}$, $H_{a b}$ and $A_b$. 
For later computations, we list the following identities 
which are entailed by \eqref{def-tE}, 
\be\label{eq-E-tE} 
\tE^a=-h^{a b} E_{b}, \AND
E_b=-\tE^a g_{ab}.
\ee 	

\begin{lemma}\label{lem-FOSHS-YM}
	The Yang--Mills system \eqref{YM-div-dyn}--\eqref{YM-bianchi-dyn} for any metric $g_{ab}$ implies the following  hyperbolic system
	\begin{align}\label{YM-FOSHS-1}
		\tilde{\mathbf{A}}^{\prime c} \nabla_{c} \p{\tE^e  \\E_{\hd}  \\ H_{\ha\hb} \\ A_{\hc}} = &
		\p{\mathfrak{D}_{1 \he}  \\  \mathfrak{D}_{2}^{d} \\ \mathfrak{D}_{3}^{a b}  \\ \mathfrak{D}_{4}^{ r} }, 
	\end{align}
	where
	\begin{align*}
		\tilde{\mathbf{A}}^{\prime c} ={} & - \p{T^c h_{\he e} & 0 & - h^{\ha c} \tensor{h}{^{\hb}_{\he}} & 0\\
			0 &  T^c h^{\hd d}  & - h^{\ha d} h^{\hb c}  & 0\\
			-  h^{a c} \tensor{h}{^b_{e}} & - h^{a \hd} h^{b c} & h^{\ha a} h^{\hb b} T^c & 0 \\
			0 & 0 & 0 & h^{r \hc} T^{c} },		
	\end{align*}
	is symmetric and
	\begin{align*}
		\tilde{\mathbf{A}}^{\prime c} T_c ={} &  \p{ h_{\he e} & 0 & 0 & 0\\
			0 &  h^{\hd d}  & 0 & 0\\
			0 & 0 & h^{\ha a} h^{\hb b} & 0 \\
			0 & 0 & 0 & h^{r \hc}  }		
	\end{align*}
	is positive,
	and 
	\begin{align*}
		\mathfrak{D}_{1 \he}( \mathbf{U}) := {} & - \widetilde{\Delta}_{1 \he}, \\
		\mathfrak{D}_{2}^{d}(\mathbf{U}) := {}& h^{d \hd} \widetilde{\Delta}_{1 \hd}, \\
		\mathfrak{D}_{3}^{a b}(\mathbf{U}):={}& - \tensor{h}{^{a \ha}} \widetilde{\Delta}_{2\ha}^{b}, \\
		\mathfrak{D}_{4}^{r}(\mathbf{U}) :={}& h^{r \ha} \widetilde{\Delta}_{3 \ha}.
	\end{align*}	
	Note that $\widetilde{\Delta}_{i}, \, i=1,2,3$ are  given in Lemma \ref{lem-maxwell-hyperbolic-0}.
\end{lemma}
\begin{remark}\label{t:ymhpsl}
	In other words, this lemma states if $(F_{ab},A_c)$ solves the Yang--Mills system \eqref{YM-div-dyn}--\eqref{YM-bianchi-dyn} for any metric $g_{ab}$, then $(\tE^e,E_{\hd},H_{\ha\hb},A_{\hc})$ defined in terms of $(F_{ab},A_c)$ in the preceding ways solves \eqref{YM-FOSHS-1}.
\end{remark}
\begin{remark}
	We observe that $\mathbf{A}^{\prime c}$ is symmetric in the following pairs of indices: $(e, \, \he)$, $(d, \, \hd)$, $(a, \, \ha)$, $(b, \, \hb)$, $(r, \, \hc)$. This property will be responsible to the symmetry of the Yang--Mills system \eqref{YM-FOSHS-1}, see Section \ref{sec-Model}. 
\end{remark}
\begin{proof}
	Keeping \eqref{eq-E-tE} in mind, we will derive a dynamic equation for $\tE^a$ through the first equation of \eqref{eq-YM-1}, which reads
	\begin{align}\label{e:eq2}
		& \tensor{h}{^a_{\ha}} T_e Q^{edc}T_d \nabla_c E_{a} + \tensor{h}{^{a}_{\ha}} T_e Q^{edc} \tensor{h}{^\hb_d} \nabla_c H_{\hb a} = \widetilde{\Delta}_{1 \ha}.
	\end{align}
	In details, we substitute \eqref{eq-E-tE} into the first term above, 
	\begin{equation*}
		\tensor{h}{^b_{\ha}} T_e Q^{edc}T_d \nabla_c E_{b} =  - \tensor{h}{^b_{\ha}} T_e Q^{edc} T_d \nabla_c (\tE^a g_{a b})   
		=   - \tensor{h}{^b_{\ha}} T_e Q^{edc} T_d g_{a b} \nabla_c \tE^a.
	\end{equation*}
	Noticing that $T_e Q^{edc} T_d = -T^c$, it then follows that 
	\begin{align}\label{e:tE-1}
		& \tensor{h}{^b_{\he}} g_{a b}  T^c \nabla_c  \tE^{a} - \tensor{h}{^b_{\he}} h^{\hb c} \nabla_c H_{\hb b}
		=  \widetilde{\Delta}_{1 \he},
	\end{align} 
	which further gives, noting that $\tensor{h}{^b_\he} g_{a b} T^c \nabla_c  \tE^{a} =  T^c h_{\he a} \nabla_c  \tE^{a}$,
	\begin{align}\label{Main-maxwell-tE-1}
		&  h_{\he e} T^c \nabla_c  \tE^{e} - \tensor{h}{^{\hb}_{\he}} h^{\ha c} \nabla_c H_{\ha \hb} = \widetilde{\Delta}_{1 \he}.
	\end{align} 
	
	Meanwhile, \eqref{e:eq2} itself is a dynamic equation for $E_{\hd}$, taking the form of
	\begin{align*}
		& \tensor{h}{^{\hd}_d} T^c \nabla_c E_{\hd}  + \tensor{h}{^{\hd}_d} h^{\hb c} \nabla_c H_{\hb \hd} = - \widetilde{\Delta}_{1 d},
	\end{align*}	
	which, after using the anti-symmetry of $H_{\ha\hb}$, is further rewritten as
	\begin{align}\label{Main-YM-E-1}
		& h^{d \hd} T^c \nabla_c E_{\hd}  - h^{d \ha} h^{\hb c} \nabla_c H_{\ha \hb} = - h^{d \hd} \widetilde{\Delta}_{1 \hd}.
	\end{align} 
	As a concluding remark, \eqref{e:eq2} is completely equivalent to \eqref{YM-div-dyn}.

	At last, we turn to the second equation of \eqref{eq-YM-1}, that is, the dynamic equation for $H_{ab}$,  
	\begin{align}\label{eq-Max-3-1}
		\tensor{h}{^{a}_{\ha}} \tensor{h}{^h_e} Q^{edc} T_d \nabla_c E_{a} +  \tensor{h}{^{c}_{\ha}} \tensor{h}{^{h a}}  \nabla_c E_{a} + \tensor{h}{^{a}_{\ha}} \tensor{h}{^h_e} Q^{edc} \tensor{h}{^{\hb}_d} \nabla_c H_{\hb a} = \widetilde{\Delta}^h_{2 \ha}.
	\end{align}
	In view of the term $- h^{\ha c} \nabla_c H_{\ha \hb}$ in \eqref{Main-maxwell-tE-1}, we have to extract a term such as $- h^{\hc c}  \nabla_c \tE^{a}$ from \eqref{eq-Max-3-1} for the sake of symmetry. For this purpose, the second term in \eqref{eq-Max-3-1} is rewritten as \[ \tensor{h}{^{c}_{\ha}} \tensor{h}{^{h a}}   \nabla_c E_{a}= -  \tensor{h}{^{c}_{\ha}} \nabla_c \tE^{h}- \tensor{h}{^{c}_{\ha}} \nabla_c h^{h a} E_{a}. \] 
	Then we multiply $h^{\hc \ha}$ on \eqref{eq-Max-3-1}, and make use of the identities (referring to Lemma \ref{lem-identity}), 
	\begin{equation*}
		h^{\hc \ha}\tensor{h}{^{\hd}_e} Q^{edc} T_d=  -h^{\hc \ha} h^{\hd c} \AND
		h^{\hc \ha} \tensor{h}{^{\hd}_e} Q^{edc} \tensor{h}{^\hb_d}= -h^{\hc \ha} T^c h^{\hb\hd}. 
	\end{equation*}		
	It follows that
	\begin{align*}
		- h^{\hc\ha} h^{\hd c}\nabla_c E_{\ha} - h^{\hc c}  \nabla_c \tE^{\hd} - & \tensor{h}{^{\hc \ha}} T^c h^{\hb\hd} \nabla_c H_{\hb\ha}   = \tensor{h}{^{\hc \ha}} \widetilde{\Delta}^{\hd}_{2\ha} + h^{\hc c}  \nabla_c h^{\hd \ha} E_{\ha}.
	\end{align*}
	Equivalently, making the rearrangement $h^{a \hb} h^{b \ha} H_{\ha\hb} = -h^{a \ha} h^{b \hb} H_{\ha \hb}$ and performing a projection, we obtain
	\begin{align}\label{Main-maxwell-H-1}
		- h^{a c} \tensor{h}{^b_e}  \nabla_c \tE^{e} - h^{a \hd} h^{b c} \nabla_c E_{\hd} + h^{a \ha} h^{\hb b} T^c \nabla_c H_{\ha \hb} ={}& \tensor{h}{^{a \ha}}  \widetilde{\Delta}^b_{2\ha}.
	\end{align}
	In fact, \eqref{Main-maxwell-H-1} is essentially equivalent to \eqref{YM-bianchi-dyn}.

	Collecting \eqref{Main-maxwell-tE-1},  \eqref{Main-YM-E-1} and \eqref{Main-maxwell-H-1} together, we achieve the following system (after changing the sign)
	\begin{align}\label{YM-main-1}
		& - \p{ T^c h_{\he e} & 0 & - h^{\ha c} \tensor{h}{^{\hb}_{\he}} \\
			0 &  T^c h^{\hd d}  & - h^{\ha d} h^{\hb c}  \\
			-  h^{a c} \tensor{h}{^b_{e}} & - h^{a \hd} h^{b c} & h^{\ha a} h^{\hb b} T^c }
		\nabla_c 
		\p{\tE^e  \\E_{\hd}  \\ H_{\ha \hb} } 
		=
		\p{ - \widetilde{\Delta}_{1 \he}  \\
			h^{d \hd} \widetilde{\Delta}_{1 \hd}  \\
			- \tensor{h}{^{a \ha}} \widetilde{\Delta}^b_{2\ha}}.
	\end{align}

	In the end,	adding the equation for $A$ \eqref{eq-YM-potential} into the above system, Lemma \ref{lem-FOSHS-YM} is then concluded.
\end{proof}

Inversely, the hyperbolic system derived in Lemma \ref{lem-FOSHS-YM} reduces to the Yang--Mills equations, if we enforce certain constraints on the data. Before verifying that, we remind ourselves $\nbar$, the induced connection on $\Sigma$, and give some practical information about $\nbar$. For any $\Sigma$-tangent tensor $\Psi$, there holds\footnote{With this notation, throughout this paper, we note that for any $\Sigma$-tangent tensor $\Psi$, $T^c \nbar_c \Psi$ vanishes, which is different from $\nbar_T \Psi$ in \eqref{def-bar-nbla-T}. } 
\begin{align}
	\nbar_c \tensor{\Psi}{^{a_1 \cdots a_s}_{b_1 \cdots b_r}} ={}&   \tensor{h}{^{\hc}_c} \tensor{h}{^{\hb_1}_{b_1}} \cdots \tensor{h}{^{\hb_r}_{b_r}} \tensor{h}{_{\ha_1}^{a_1}} \cdots \tensor{h}{_{\ha_s}^{a_s}} \nabla_{\hc} \tensor{\Psi}{^{\ha_1 \cdots \ha_s}_{\hb_1 \cdots \hb_r}}, \label{def-bar-nbla}
	\intertext{and}
	\nbar_T \tensor{\Psi}{^{a_1 \cdots a_s}_{b_1 \cdots b_r}} ={}& \tensor{h}{^{\hb_1}_{b_1}} \cdots \tensor{h}{^{\hb_r}_{b_r}} \tensor{h}{_{\ha_1}^{a_1}} \cdots \tensor{h}{_{\ha_s}^{a_s}} \nabla_T \tensor{\Psi}{^{\ha_1 \cdots \ha_s}_{\hb_1 \cdots \hb_r}}.  \label{def-bar-nbla-T}
\end{align}	
In particular, we note that, for a $(0, r)$ tensor $\Psi$ that is $\Sigma$-tangent, both of $\nbar_T \Psi$ and $\lie_T \Psi$ are $\Sigma$-tangent, and they are related as follows,
\begin{equation}\label{e:lienb}
	\lie_T \Psi_{b_1 \cdots b_r} = \nbar_T \Psi_{b_1 \cdots b_r} + \sum_i \nbar_{b_i} T^b \Psi_{b_1 \cdots b \cdots b_r}. 
\end{equation}
As a remark, $E_{d}$, $H_{a b}$, $A_b$ and their Lie derivatives $\lie_T A_a, \, \lie_T E_b, \, \lie_T H_{a b}$ are all $\Sigma$-tangent. 


The following lemma plays an important role in the proof of Theorem \ref{t:depctrnt} that states the equivalence between the EYM equations and the target symmetric hyperbolic system. 
\begin{lemma}\label{t:ctrtevl}
	Given a metric $g$, if $(\tE^e, \, E_d, \, H_{a b}, \, A_c)$ solves the hyperbolic system \eqref{YM-FOSHS-1} for $t\in[0,T_\star)$ and the initial data obeys the following constraints\footnote{These constraints include the Yang--Mills constraints \eqref{constraint-YM-1} and two extra constraints \eqref{id-tE-E}-\eqref{id-A-H} which fix the relations between $\tE^e$ and $E_{d}$, $A_a$ and $H_{a b}$ respectively. 
	}: 
	\begin{align}
		\bigl(\tE^e \tensor{h}{^d_e} + h^{\hd d} E_{\hd}\bigr)|_{t=0} ={}&0, \label{id-tE-E} \\
		\bigl(\nbar_a A_{ b} - \nbar_b A_{ a} + [A_{ a}, A_{  b}]-H_{ a b}\bigr)|_{t=0} ={}& 0, \label{id-A-H} \\
		\bigl(\nbar^a E_{a} + h^{a b} [A_{a}, E_{b}]\bigr)|_{t=0} = {}&0, \label{constraint-YM-1} 
	\end{align} 
	then these constraints holds for any future time $t\in[0,T_\star)$, 
	\begin{align}
		\tE^e \tensor{h}{^d_e} + h^{\hd d} E_{\hd} \equiv {}&0, \label{id-tE-E2} \\
		\nbar_a A_{b} - \nbar_b A_{a} + [A_{a}, A_{b}] \equiv {}& H_{a b}, \label{id-A-H2} \\
		\nbar^a E_{a} + h^{a b} [A_{a}, E_{b}] \equiv {}&0. \label{constraint-YM-2} 
	\end{align}
\end{lemma}

\begin{proof}
	To simplify the presentation of the proof, we define 
	\begin{align}
		\alpha^d :={}& \tE^e \tensor{h}{^d_e} + h^{\hd d} E_{\hd}, \label{e:a}\\
		\beta_{p q}: ={}& \nbar_p A_{q} - \nbar_q A_p + [A_p, A_q] - H_{p q}, \label{e:b}\\
		\gamma :={}& \nbar^a E_a + h^{a b} [A_a, E_b]\label{e:g}.
	\end{align}		
	Let us denote the second fundamental form $k_{a b}$ and the relevant quantities by,  \[k_{a b} := \frac{1}{2} \lie_T h_{a b}, \quad \tr k := h^{a b} k_{a b} \AND \mathbf{k}_{a b} := \frac{1}{2} \lie_T g_{a b}.\] 
	Then, there follows the identities, \[k_{a b} = \nbar_a T_b = \nbar_b T_a, \quad \tr k =  \nbar^a T_a = h^{a b} \nbar_a T_b \AND \mathbf{k}_{a b} = k_{a b} - \frac{1}{2} T_a \nabla_T T_b - \frac{1}{2} T_b \nabla_T T_a.\] 
	We also denote
	$	\mathbb{A}_a$, $	\mathbb{B}^{\he}$, $	\mathbb{D}_d$ and $	\mathbb{C}_{p q}$ 	by
	\begin{align*}
		\mathbb{A}_a : ={}& \nbar_{T} A_a + E_a + \nbar_{a} T^{d}  A_{d}, 
		\\
		\mathbb{B}^{\he} : ={}& \nbar_T  \tE^{\he} - h^{\he e} \nbar^{\ha} H_{\ha e} - \nbar^a T_a h^{\he e} E_{e} - h^{\he e} h^{c d} [A_d, H_{c e}] + \nbar^d T^{\he} E_d - h^{\he f} \nabla_T T^e H_{e f}, 
		\\
		\mathbb{D}_d :={}& \nbar_T E_d + \nbar^a H_{a d} + \nbar^aT_a E_d - \nbar^a T_d  E_a + \nbar_T T^a H_{a d} +  h^{a b} [A_a, H_{b d}],  
		\\
		\mathbb{C}_{pq} :={}& \nbar_T H_{p q}  +  \nbar_p E_q - \nbar_q E_p -  h_{q d} \nbar_p \alpha^d + \nbar_p T^b H_{b q} + \nbar_q T^b H_{p b}  \nnb \\
		&- \nabla_T T_q E_p + \nabla_T T_p E_q + [A_p, E_q] - [A_q, E_p], 
	\end{align*}		
	which can be equivalently expressed, in terms of Lie derivatives, as 
	\begin{align}
		\mathbb{A}_a ={}& \lie_{T} A_a + E_a,\label{e:A2} \\
		\mathbb{D}_d ={}& \lie_T E_d + \nbar^a H_{a d} + \tr k E_d - 2 \tensor{k}{^a_d} E_a + \nbar_T T^a H_{a d} +  h^{a b} [A_a, H_{b d}], \label{e:D2}\\
		\mathbb{C}_{pq} ={}& \lie_T H_{p q}  +  \nbar_p E_q - \nbar_q E_p -  h_{q d} \nbar_p \alpha^d - \nabla_T T_q E_p + \nabla_T T_p E_q + [A_p, E_q] - [A_q, E_p]. \label{e:C2}
	\end{align}
	In terms of these definitions and notations, the system \eqref{YM-FOSHS-1} reads, respectively, as 
	\begin{align}
		\mathbb{A}_a ={}&0, \label{YM-eq-A-1} \\
		\mathbb{B}^{\he}={}& 0, \label{YM-tE-1} \\
		\mathbb{D}_d ={}&0, \label{YM-prop-1} \\
		\mathbb{C}_{pq} ={}&0. \label{eq-E-H}
	\end{align}  
	
	Next, we proceed with deducing several geometric identities that concern the dynamics of $\alpha^d$, $\beta_{pq}$ and $\gamma$. With the help of  \eqref{YM-eq-A-1}--\eqref{eq-E-H} (i.e. the hyperbolic system \eqref{YM-FOSHS-1}), these identities reduce to a closed and homogeneous system on $\alpha^d$, $\beta_{pq}$ and $\gamma$, which further implies that if  $\alpha^d=0$, $\beta_{pq}=0$ and $\gamma=0$ initially, then $\alpha^d=0$, $\beta_{pq}=0$ and $\gamma=0$ for all $t\in[0,T_\star)$.
	
	\underline{$(1)$ The equation of $\alpha^d$:} Using $\nbar_T$ to act on \eqref{e:a} and noting that $\nbar_c \tensor{h}{^d_e}=0$ and $\nbar_T h^{\hd d}=0$), lead to 
	\begin{align}\label{eq-T-alpha}
		\nbar_T \alpha^d ={}& \mathbb{B}^d + \mathbb{D}^d + \nbar_T \tensor{h}{^d_e} \tE^e +  \nbar_T h^{\hd d} E_{\hd}  
		=  \mathbb{B}^d + \mathbb{D}^d.
	\end{align}

	\underline{$(2)$ The equation of $\beta_{pq}$:} We note that
	\[ \lie_T \di A = \di \lie_T A  = \di \mathbb{A} - \di E. \] That is, 
	\begin{align} \label{e:hhl}
		\tensor{h}{^a_p} \tensor{h}{^b_q} \lie_T (\di A)_{a b} = \tensor{h}{^a_p} \tensor{h}{^b_q} (\di \mathbb{A})_{a b} - \nbar_p E_q + \nbar_q E_p.
	\end{align}
	Meanwhile, noting
	\begin{equation*}
		\lie_T \tensor{h}{^a_p} = T_p  \nabla_T T^a + T^a\nabla_T T_p -\nbar_p T^a+ \tensor{h}{^a_c} \nabla_p T^c =T^a\nabla_T T_p, 
	\end{equation*}
	then the term on the left hand side of \eqref{e:hhl} is related to $\lie_T (\nbar_p A_q - \nbar_q A_p)$ as follows,
	\begin{align}\label{e:hhl2}
		\tensor{h}{^a_p} \tensor{h}{^b_q} \lie_T (\di A)_{a b} = {}&  \lie_T (\tensor{h}{^a_p} \tensor{h}{^b_q} \di A)_{a b}-\tensor{h}{^b_q}\lie_T \tensor{h}{^a_p}  (\di A)_{a b}-\tensor{h}{^a_p}\lie_T  \tensor{h}{^b_q} (\di A)_{a b} \notag\\
		= {}& \lie_T (\nbar_p A_q - \nbar_q A_p)  - \nabla_T T_q E_p + \nabla_T T_p E_q  + \nabla_T T_q \mathbb{A}_p - \nabla_T T_p \mathbb{A}_q.
	\end{align}
	On the other hand, thanks to Lemma \ref{t:lieleb} (in Appendix \ref{s:lieleb}), there is the identity
	\begin{align}\label{e:llie}
		\lie_T ([A, A])_{p q} = {}& [\mathbb{A}_p, A_q] + [A_p, \mathbb{A}_q] + [-E_p, A_q] + [A_p, - E_q].
	\end{align}	
	Hence, by \eqref{e:hhl}, \eqref{e:hhl2} and \eqref{e:llie}, we obtain 
	\begin{align*}
		& \lie_T (\nbar_p A_q - \nbar_q A_p + [A, A]_{p q}) \\
		= &- \nbar_p E_q + \nbar_q E_p + [-E_p, A_q] + [A_p, - E_q] + \nabla_T T_q E_p - \nabla_T T_p E_q\\
		& +  \tensor{h}{^a_p} \tensor{h}{^b_q} (\di \mathbb{A})_{a b}  - \nabla_T T_q \mathbb{A}_p + \nabla_T T_p \mathbb{A}_q +[\mathbb{A}_p, A_q] + [A_p, \mathbb{A}_q].
	\end{align*}
	With the help of the definition of $\mathbb{C}_{p q}$ \eqref{e:C2}, $\lie_T H_{p q}$ naturally enters into the above formula,
	\begin{align*}
		\lie_T (\nbar_p A_q - \nbar_q A_p + [A, A])_{p q} 
		={}& \lie_T H_{p q} - \mathbb{C}_{p q} -  h_{q d} \nbar_p \alpha^d + \tensor{h}{^a_p} \tensor{h}{^b_q} (\di \mathbb{A})_{a b} \nnb \\
		& - \nabla_T T_q \mathbb{A}_p + \nabla_T T_p \mathbb{A}_q + [\mathbb{A}_p, A_q] + [A_p, \mathbb{A}_q].
	\end{align*}
	It then becomes the aiming identity on $\beta_{pq}$, for using the definition \eqref{e:b}, we arrive at
	\begin{align}
		\lie_T \beta_{p q} 
		={}& - \mathbb{C}_{p q} -  h_{q d} \nbar_p \alpha^d + \tensor{h}{^a_p} \tensor{h}{^b_q} (\di \mathbb{A})_{a b} \nnb \\
		& - \nabla_T T_q \mathbb{A}_p + \nabla_T T_p \mathbb{A}_q + [\mathbb{A}_p, A_q] + [A_p, \mathbb{A}_q]. \label{eq-beta-2}
	\end{align}

	\underline{$(3)$ The equation of $\gamma$:} It follows by straightforward calculations that
	\begin{align*}
		\lie_T \gamma ={}& h^{a b} \lie_T \nabla_a E_b  + \lie_T h^{a b} (\nabla_a E_b + [A_a, E_b] ) + h^{a b}( [\lie_T A_a, E_b] + [A_a, \lie_T E_b]) \\
		= {}& h^{a b} \nabla_b  \lie_T E_a -2\nbar^a \tensor{k}{^b_{a}} E_b +\nbar^b \tr k E_b + \tr k \nabla_T T^c E_c - 2 k^{a b} \nbar_b E_a - 2 k^{a b} [A_a, E_b]  \\
		&+ \nabla_T T^a \nabla_T E_a - \nabla_T T^b \tensor{k}{^a_b} E_a + h^{a b} [\mathbb{A}_a, E_b] + h^{a b} [A_a, \lie_T E_b], 
	\end{align*}
	where we use the commuting identity between $\lie_T$ and $\nabla$,
	\begin{align*}
		h^{a b} \lie_T \nabla_a E_b ={}& h^{a b} \nabla_b  \lie_T E_a - h^{a b} (\nabla_a \tensor{\mathbf{k}}{^c_b} + \nabla_b \tensor{\mathbf{k}}{^c_a} - \nabla^c \mathbf{k}_{a b}) E_c \\
		={}& h^{a b} \nabla_b  \lie_T E_a -2\nbar^a \tensor{k}{^{ b}_a} E_b +\nbar^b \tr k E_b + \tr k \nabla_T T^c E_c. 
	\end{align*}
	Making use of $\mathbb{D}_a$ \eqref{e:D2} to substitute the $\lie_T E_a$ and $\nabla_T E_a$ above leads to
	\begin{align}\label{e:ltg}
		\lie_T \gamma = {}& \nbar^a \mathbb{D}_a + \nabla_T T^a \mathbb{D}_a + h^{a b} [\mathbb{A}_a, E_b] + h^{a b} [A_a, \mathbb{D}_b]- \tr k \gamma + \nbar^a \nbar^b H_{a b} \notag \\
		&  - \nbar^a(\nabla_T T^b) H_{b a}  -[\nbar^a A^b, H_{b a}] - h^{a b} h^{\ha \hb} [A_a,[A_{\ha}, H_{\hb b}]]. 
	\end{align}

	In the following, we devote to simplifying \eqref{e:ltg}. We first verify $\nbar^a \nbar^b H_{ab}=0$ and \\
	$\nbar^a(\nabla_T T^b) H_{b a}=0$. This is confirmed by the identities
	\begin{align}\label{e:t10}
		\nbar^a \nbar^b H_{ab} =   \frac{1}{2} (\nbar^a \nbar^b H_{ab} - \nbar^b \nbar^a H_{ab} )  
		=  \frac{1}{2} (\bar R^{c b} H_{b c} - \bar R^{a c}H_{a c}) = 0,
	\end{align}
	and 
	\begin{align}\label{e:t20}
		&\nbar^a(\nabla_T T^b) H_{b a} = \tensor{k}{^a_c} k^{c b} H_{b a} + T^c \nabla^a \nabla_c T^b H_{b a} = T^c \nabla^a \nabla_c T^b H_{b a} \notag  \\
		& \hspace{1cm} = \nabla_T \nabla^a T^b H_{b a} + R_{aT b T} H^{a b} =\nbar_T k^{a b} H_{b a}-\nabla_T T^a \nabla_T T^b H_{ba} + R_{aT b T} H^{a b} =0,
	\end{align}	
	where the anti-symmetry of $H_{a b}$ is used and $\bar R_{ab}$ denotes the Ricci curvature associated to $\nbar$.

	In the next, let us justify the following identity,
	\begin{equation}\label{e:t50}
		-[\nbar^a A^b, H_{b a}] - h^{a b} h^{\ha \hb} [A_a,[A_{\ha}, H_{\hb b}]]=\frac{1}{2} [ \beta^{a b}, H_{a b}]. 
	\end{equation}
	We note, by the definition of $\beta^{ab}$ \eqref{e:b}  and using $[H^{ba},H_{ba}]\equiv 0$, the following identity, 
	\begin{align}\label{e:t30}
		& -[\nbar^a A^b, H_{b a}] = -\frac{1}{2} [\nbar^a A^b - \nbar^b A^a, H_{b a}] \notag  \\
		& \hspace{1cm} = -\frac{1}{2} [H^{a b} - [A^a, A^b] + \beta^{a b}, H_{b a}] 
		= \frac{1}{2} [H_{a b}, [A^a, A^b]] + \frac{1}{2} [ \beta^{a b}, H_{a b}]. 
	\end{align}
	Moreover, using the Jacobi identity, we obtain 
	\begin{align*}
		& \frac{1}{2} [H_{a b}, [A^a, A^b]] -[A^b, [A^{a}, H_{a b}]] \\
		={}&- \frac{1}{2} [A^a, [A^b, H_{a b}]] - \frac{1}{2} [A^b, [H_{a b}, A^a]]   -[A^b, [A^{a}, H_{a b}]] \\
		={}& \frac{1}{2} [A^b, [A^a, H_{a b}]] + \frac{1}{2} [A^b, [A^a, H_{a b}]]   -[A^b, [A^{a}, H_{a b}]] =0, 
	\end{align*}
	which, in turn, yields
	\begin{equation}\label{e:t40}
		[A^b, [A^{a}, H_{a b}]]=\frac{1}{2} [H_{a b}, [A^a, A^b]]. 
	\end{equation}	
	Then, the identities \eqref{e:t30} and \eqref{e:t40}, with the help of the temporal gauge \eqref{temporal}, imply \eqref{e:t50}. 
	
	Finally, inserting  \eqref{e:t10}, \eqref{e:t20} and \eqref{e:t50} into \eqref{e:ltg}, we arrive at
	\be\label{eq-T-gamma}
	\lie_T \gamma + \tr k \gamma = \nbar^a \mathbb{D}_a + \nabla_T T^a \mathbb{D}_a + h^{a b} [\mathbb{A}_a, E_b] + h^{a b} [A_a, \mathbb{D}_b] + \frac{1}{2} [ \beta^{a b}, H_{a b}].
	\ee
	
	Collecting \eqref{eq-T-alpha}, \eqref{eq-beta-2} and \eqref{eq-T-gamma} together, we express the propagating equations of $\alpha^d$, $\beta_{pq}$ and $\gamma$ in terms of $\mathbb{A}_a$, $\mathbb{B}^{\he}$, $\mathbb{D}_d$ and $\mathbb{C}_{pq}$, 
	\begin{align*}
		\bar \nabla_T \alpha^d ={}& \mathbb{B}^d + \mathbb{D}^d, \\
		\lie_T \beta_{p q} ={}&- \mathbb{C}_{p q} -  h_{q d} \nbar_p \alpha^d +\tensor{h}{^a_p} \tensor{h}{^b_q} (\di \mathbb{A})_{a b} - \nabla_T T_q \mathbb{A}_p + \nabla_T T_p \mathbb{A}_q + [\mathbb{A}_p, A_q] + [A_p, \mathbb{A}_q],  \\
		\lie_T \gamma + \tr k \gamma ={}& \nbar^a \mathbb{D}_a + \nbar_T T^a \mathbb{D}_a + h^{a b} [\mathbb{A}_a, E_b] + h^{a b} [A_a, \mathbb{D}_b] + \frac{1}{2} [ \beta^{a b}, H_{a b}].
	\end{align*}		
	Then, using \eqref{YM-eq-A-1}--\eqref{eq-E-H} (i.e., \eqref{YM-FOSHS-1}), we derive for $\alpha^d$, $\beta_{pq}$ and $\gamma$ a system of homogeneous equation, 
	\begin{align} 
		\bar \nabla_T \alpha^d ={}& 0,
		\label{e:aode}\\
		\lie_T \beta_{p q} ={}& - h_{q d} \nbar_p \alpha^d, \label{e:bode}\\
		\lie_T \gamma ={}& - \tr k \gamma +  \frac{1}{2} [ \beta^{a b}, H_{a b}], \label{e:gode}
	\end{align}
	with initial data $\alpha^d|_{t=0}=0$, $\beta_{pq}|_{t=0}=0$ and $\gamma|_{t=0}=0$. It follows that $\alpha^d, \, \beta_{pq}, \, \gamma$ all vanish during the propagation. The proof can be carried out by various approaches\footnote{An alternative method is to construct a local \textit{adapted coordinate system} (see \cite[\S $2.5$]{Gourgoulhon2012} for details) to the vector field $T^a$, then solve these equations in this local coordinates due to the fact that the Lie derivative reduces to a coordinate derivative along the $T^a$-integral curves. In fact, by virtue of the definition of Lie derivatives, we know that if the data vanish, then the variable vanishes along the integral curves. }, for instance, a hierarchy of energy estimates that begins with $\alpha^d$, then $\beta_{p q}$, and ends up with $\gamma$ will complete the proof. 
\end{proof}

In the end, adapted to the target formulation of hyperbolic system, \eqref{YM-FOSHS-1} is transformed into the following form.  
\begin{corollary}\label{thm-FOSHS-YM}
	Given any metric $g$, the Yang--Mills system \eqref{YM-div-dyn}--\eqref{YM-bianchi-dyn} implies the following propagation equations
	\begin{align}\label{YM-FOSHS}
		\tilde{\mathbf{A}}^{ c } \nb_{c} \p{\tE^{\check{e}}  \\E_{\check{d}}  \\ H_{\check{a}\check{b}} \\ A_{\check{s}}} =  -(\tilde{\mathbf{A}}^{\hc} \nu_{\hc})\nu^c\nb_c \p{\tE^{\check{e}}  \\E_{\check{d}}  \\ H_{\check{a}\check{b}} \\ A_{\check{s}}} +\tilde{\mathbf{A}}^{\hc} \tensor{\ulh}{^c_{\hc}} \nb_c \p{\tE^{\check{e}}  \\E_{\check{d}}  \\ H_{\check{a}\check{b}} \\ A_{\check{s}}} = \tilde G (\mathbf{U}), 
	\end{align}
	where
	\begin{align*}
		\tilde{\mathbf{A}}^{ c} ={} & - \p{T^c h_{\he e}\ulh^{b \he} \tensor{\ulh}{^{e}_{\check{e}}} & 0 & - h^{\ha c} \tensor{h}{^{\hb}_{\he}}\ulh^{b \he}\tensor{\ulh}{^{\check{a}}_{\ha}} \tensor{\ulh}{^{\check{b}}_{\hb}}  & 0\\
			0 &  T^c h^{\hd d} \ulh_{d f} \tensor{\ulh}{^{\check{d}}_{\hd}}  & - h^{\ha d} h^{\hb c} \ulh_{d f} \tensor{\ulh}{^{\check{a}}_{\ha}} \tensor{\ulh}{^{\check{b}}_{\hb}}  & 0\\
			-  h^{a c} \tensor{h}{^b_{e}} \ulh_{a \bar a} \ulh_{b \bar b} \tensor{\ulh}{^{e}_{\check{e}}} & - h^{a \hd} h^{b c} \ulh_{a \bar a} \ulh_{b \bar b} \tensor{\ulh}{^{\check{d}}_{\hd}} & h^{\ha a} h^{\hb b} T^c \ulh_{a \bar a} \ulh_{b \bar b} \tensor{\ulh}{^{\check{a}}_{\ha}} \tensor{\ulh}{^{\check{b}}_{\hb}}  & 0 \\
			0 & 0 & 0 & h^{r \hc} T^{c} \ulh_{\check{c} r} \tensor{\ulh}{^{\check{s}}_{\hc}}  },		
	\end{align*} 
	is symmetric and $\tilde{\mathbf{A}}^{c} \nu_c$ is positive, 
	and 
	\begin{align*}
		\tilde G (\mathbf{U}) =  
		-\diag \{\ulh^{b \he}, \ulh_{d f}, \ulh_{a \bar a} \ulh_{b \bar b}, \ulh_{\check{c} r} \} \tilde{\mathbf{A}}^{\prime c} \p{ \tensor{X}{^{e}_{c d}} \tE^d \\ - \tensor{X}{^d_{c \hd }} E_d \\ - \tensor{X}{^d_{c \ha }} H_{d \hb } - \tensor{X}{^d_{c \hb } } H_{\ha d}  \\ - \tensor{X}{^d_{c \hc}} A_{d} 
		} + 
		\p{ \ulh^{b \he} \mathfrak{D}_{1 \he}(\mathbf{U}) \\
			\ulh_{d f}\mathfrak{D}_{2}^{d}(\mathbf{U}) \\
			\ulh_{a \bar a} \ulh_{b \bar b} \mathfrak{D}_{3}^{a b}(\mathbf{U}) \\
			\ulh_{\check{c} r} \mathfrak{D}_{4}^{r}(\mathbf{U}) }.
	\end{align*}
	Here we recall that $\tensor{X}{^e_{ca}}$ is defined by \eqref{E:X2}, $\tilde{\mathbf{A}}^{\prime c}$  and $\mathfrak{D}_{i}(\mathbf{U}), \, i =1,\cdots,4$ are given in Lemma \ref{lem-FOSHS-YM}.
\end{corollary}
\begin{proof}
	We exhibit the difference between the two covariant derivatives $\nabla_c$ and $\nb_c$ by the $\tensor{X}{^a_{b c}}$ terms and its application to the principle part in the hyperbolic system \eqref{YM-FOSHS-1} (see Lemma \ref{lem-FOSHS-YM}),
	\[\tilde{\mathbf{A}}^{\prime c} \nabla_c  \p{\tE^e  \\E_{\hd}  \\ H_{\ha\hb} \\ A_{\hc}} =  \tilde{\mathbf{A}}^{\prime c} \nb_c  \p{\tE^e  \\E_{\hd}  \\ H_{\ha\hb} \\ A_{\hc}} + \tilde{\mathbf{A}}^{\prime c} \p{ \tensor{X}{^{e}_{c d}} \tE^d \\ - \tensor{X}{^d_{c \hd }} E_d \\ - \tensor{X}{^d_{c \ha }} H_{d \hb } - \tensor{X}{^d_{c \hb } } H_{\ha d}  \\ - \tensor{X}{^d_{c \hc}} A_{d} }, \]
	and act the matrix $\diag \{\ulh^{b \he}, \, \ulh_{d f}, \, \ulh_{a \bar a} \ulh_{b \bar b}, \, \ulh_{\check{c} r} \}$ on the resulted equation. Then due to \eqref{E:NBNU} and the facts that $\tensor{h}{^{\hd}_d} \nu_{\hd} = 0$, $\nu_{\check{e}} \tE^{\check{e}} = 0$\footnote{In fact, if $\xi^a \tensor{h}{^b_a} =\xi^b$, then $\xi^a \nu_a=0$. This can be proved by the following argument. $\xi^a \tensor{h}{^b_a}=\xi^b$ implies $\xi^a T_a=0$, which, due to $T_c = |\lambda|^{-\frac{1}{2}} \nu_c$, in turn yields $\xi^a \nu_a=0$. }, there are the identities for the principle part, such as,
	\als{
		T^c h_{\he e}\ulh^{b \he} \nb_c \tE^{e} ={}&  T^c h_{\he e}\ulh^{b \he} \nb_c \left( (\tensor{\ulh}{^{e}_{\check{e}}} - \nu^e \nu_{\check{e}}) \tE^{\check{e}} \right) \\
		={}& T^c h_{\he e}\ulh^{b \he} \tensor{\ulh}{^{e}_{\check{e}}} \nb_c \tE^{\check{e}}, \\
		T^c h^{\hd d} \ulh_{d f} \nb_c E_{\hd}  ={}& T^c h^{\hd d} \ulh_{d f} \nb_c \left( ( \tensor{\ulh}{^{\check{d}}_{\hd}} - \nu_{\hd} \nu^{\check{d}})  E_{\check{d}}  \right) \\
		={}&  T^c h^{\hd d} \ulh_{d f} \tensor{\ulh}{^{\check{d}}_{\hd}} \nb_c  E_{\check{d}},
	}
	and analogous identities hold true for the rest entries. Therefore, we are able to rewrite the principle part as (the definition of $\tilde{\mathbf{A}}^{ c }$ is given in the corollary)
	\[\diag \{\ulh^{b \he}, \,\ulh_{d f},\, \ulh_{a \bar a} \ulh_{b \bar b}, \,\ulh_{\check{c} r} \}  \tilde{\mathbf{A}}^{\prime c} \nb_c  \p{\tE^e  \\E_{\hd}  \\ H_{\ha\hb} \\ A_{\hc}} = \tilde{\mathbf{A}}^{ c } \nb_{c} \p{\tE^{\check{e}}  \\E_{\check{d}}  \\ H_{\check{a}\check{b}} \\ A_{\check{s}}}, \] and obtain \eqref{YM-FOSHS}. 	
\end{proof}

\subsection{Equivalence between the EYM equations and the hyperbolic system}\label{s:equiv}	
Let us gather the reduced Einstein equations  \eqref{E:Einstein-sym-thm2} and the Yang--Mills equations \eqref{YM-FOSHS} together (see Corollary \ref{t:Ein2} and \ref{thm-FOSHS-YM}), then obtain a symmetric hyperbolic formulation of the Einstein--Yang--Mills system, 
\begin{align}\label{e:EYMmain}
	- \mathbf{A}^0 \nu^c \nb_c \mathbf{U} + \mathbf{A}^c \tensor{\ulh}{^	b_c} \nb_b \mathbf{U} =  G (\mathbf{U}),
\end{align}
where we recall $\mathbf{U}$ is the collection of all variables defined by \eqref{def-U}, 
\begin{align*}
	\mathbf{A}^0:=\p{ \check{\mathbf{A}}^0 & 0 \\ 0 & \tilde{\mathbf{A}}^{\hc} \nu_{\hc}}, \quad \mathbf{A}^c \tensor{\ulh}{^b_c} :=\p{ \check{\mathbf{A}}^c \tensor{\ulh}{^b_c} & 0 \\ 0 & \tilde{\mathbf{A}}^{\hc} \tensor{\ulh}{^b_{\hc}} }  \AND G(\mathbf{U}):=\p{\check G (\mathbf{U})\\\tilde G (\mathbf{U}) }. 
\end{align*}

In what follows, we address the \textit{equivalence} between the specified Cauchy problems of the above symmetric hyperbolic system \eqref{e:EYMmain} and the EYM system \eqref{eq-einstein}--\eqref{eq-YM-bianchi-1}. By ``\textit{equivalence}'' in this article, we mean that, under the wave gauge \eqref{E:CONSTR1} and the temporal gauge \eqref{temporal}, if $(g, A)$ solves the Cauchy problem of the EYM equations \eqref{eq-einstein} and  \eqref{eq-YM-div}--\eqref{eq-YM-Bianchi} with the Einstein and Yang--Mills constraints, then $\mathbf{U}$ determined by $(g, A)$ via the definitions \eqref{vari-metric}--\eqref{def-U} solves the hyperbolic system \eqref{e:EYMmain}. Conversely, if $\mathbf{U}$ solves the Cauchy problem of the hyperbolic system \eqref{e:EYMmain} with the specified constraints (to be specified in Theorem \ref{t:depctrnt}), then we can construct from $\mathbf{U}$ a solution $(g, A)$ that solves the reduced Einstein--Yang--Mills equation \eqref{eq-einstein} and \eqref{eq-YM-div}--\eqref{eq-YM-Bianchi} under the wave gauge \eqref{E:CONSTR1} and the temporal gauge \eqref{temporal}.

\begin{theorem}\label{t:depctrnt}
	We consider the Cauchy problem of the hyperbolic system \eqref{e:EYMmain} for $\mathbf{U}$ with data obeying the Einstein constraints, the wave gauge constraints and the constraints in Lemma \ref{t:ctrtevl} (abbreviated as the Cauchy problem of the hyperbolic system), and at the same time the Cauchy problem of the EYM equations \eqref{eq-einstein} and \eqref{eq-YM-div}--\eqref{eq-YM-Bianchi} under the wave gauge \eqref{E:CONSTR1} and the temporal gauge \eqref{temporal} with data satisfying the Einstein and Yang--Mills constraints (abbreviated as the Cauchy problem of the EYM system). Then the Cauchy problems of these two systems are equivalent in the existence region of solutions $t\in[0,T_\star)$. 
\end{theorem}

\begin{proof}
	\underline{$(1)$ EYM equations $\Rightarrow$ \eqref{e:EYMmain}:} Theorem \ref{thm-hyperb-gravity},  Lemma \ref{lem-FOSHS-YM} and the corresponding proofs immediately tell that if $(g, A)$ solves the Cauchy problem of the EYM equations \eqref{eq-einstein} and \eqref{eq-YM-div}--\eqref{eq-YM-Bianchi} under the wave gauge \eqref{E:CONSTR1} and the temporal gauge \eqref{temporal}, then $\mathbf{U}$ determined by $(g, A)$ via the definitions \eqref{vari-metric}--\eqref{def-U} solves the hyperbolic system \eqref{e:EYMmain}.

	\underline{$(2)$ \eqref{e:EYMmain} $\Rightarrow$ EYM equations:} Let $\mathbf{U}$ solve the Cauchy problem of the hyperbolic system \eqref{e:EYMmain} and it naturally gives the metric $g$ and the potential $A$ (with $A_T = 0$).  
	If we can verify $(g, A)$ solves the EYM equation \eqref{eq-einstein} and \eqref{eq-YM-div}--\eqref{eq-YM-Bianchi} under the wave gauge \eqref{E:CONSTR1} and temporal gauge \eqref{temporal}, then we complete the proof. We recall that, by Lemma \ref{t:ctrtevl}, the variables $\tE^e$, $E_c$, $H_{a b}$ in $\mathbf{U}$ obey the following identities, for  $t\in[0,T_\star)$,
	\begin{align}
		\tE^e \tensor{h}{^d_e} + h^{\hd d} E_{\hd} \equiv {}&0, \label{id-tE-E2a} \\
		\nbar_a A_{b} - \nbar_b A_{ a} + [A_{a}, A_{b}] \equiv {}& H_{a b}, \label{id-A-H2a} \\
		\nbar^a E_{a} + h^{a b} [A_{a}, E_{b}] \equiv {}&0. \label{constraint-YM-2a}
	\end{align} 
	In this step, we work backwards on the original hyperbolic formulation of Einstein--Yang--Mills system \eqref{YM-FOSHS-1} and \eqref{E:Einstein-sym-thm}, which amounts to the hyperbolic system \eqref{e:EYMmain}.	
	
	\underline{$(2a)$ \eqref{YM-FOSHS} (or \eqref{YM-FOSHS-1}) $\Rightarrow$ Yang--Mills equations in temporal gauge:}  We remark that the Yang--Mills equations in temporal gauge involve \eqref{eq-YM-div}--\eqref{eq-YM-Bianchi} with the connection fulfilling $A_a T^a =0$ and the curvature $F_{a b}$ given in terms of $A$ by \eqref{def-F}. Moreover, \eqref{eq-YM-div}--\eqref{eq-YM-Bianchi} compose of the propagation equations \eqref{YM-div-dyn}--\eqref{YM-bianchi-dyn}, and the constraints \eqref{YM-div-F-const}--\eqref{YM-bianchi-F-const}. All of them can be recovered by the hyperbolic system \eqref{YM-FOSHS-1} together with the constraints \eqref{id-tE-E2a}--\eqref{constraint-YM-2a} that are maintained during the evolution. 
	
	First of all, let $A$ be the potential induced from the solution $\mathbf{U}$ which satisfies $A_a T^a =0$ and define the associated curvature $F_{a b}$ through \eqref{def-F}. By this definition, $F_{a b} \tensor{h}{^a_p} \tensor{h}{^b_q} = \nbar_a A_{b} - \nbar_b A_{ a} + [A_{a}, A_{b}]$, then the Yang--Mills Bianchi constraint equation \eqref{YM-bianchi-F-const} automatically holds true.
	In addition, as the definition of $F_{a b}$ \eqref{def-F} implies, the identity \eqref{id-A-H2a} and the dynamic equation for $A$, namely the fourth line in the hyperbolic system \eqref{YM-FOSHS-1} entail that 
	\begin{equation}\label{Recover-F-H-E}
		F_{a b} \tensor{h}{^a_p} \tensor{h}{^b_q} = H_{p q}, \quad F_{a b} T^b \tensor{h}{^a_c} = E_c.
	\end{equation}
	With the help of \eqref{Recover-F-H-E}, the second equation of \eqref{YM-FOSHS-1} recovers the Yang--Mills propagation equation \eqref{YM-div-dyn} (see Lemma \ref{lem-FOSHS-YM}), and \eqref{constraint-YM-2a} becomes exactly the Yang--Mills constraint \eqref{YM-div-F-const}.
	Moreover, \eqref{Recover-F-H-E} and the first relation \eqref{id-tE-E2a} imply that the third equation of \eqref{YM-FOSHS-1} recovers the Yang--Mills propagation equation \eqref{YM-bianchi-dyn} as well (see Lemma \ref{lem-FOSHS-YM}). 
	Finally, the remaining equation in \eqref{YM-FOSHS-1}, i.e. the first one, is identical to the second one in \eqref{YM-FOSHS-1} due to the relation \eqref{id-tE-E2a}. In other words, the hyperbolic system \eqref{YM-FOSHS-1} together with the constraints \eqref{id-tE-E2a}-\eqref{constraint-YM-2a} contains precisely the information of the Yang--Mills equations. 
	
	\underline{$(2b)$ \eqref{E:Einstein-sym-thm2} (or \eqref{E:Einstein-sym-thm}) $\Rightarrow$  Einstein equations:}  
	It follows from the evolution of the wave gauge (if initially the wave gauge holds, it holds in the developments. The details can be found in \cite[\S$14.2$]{Ringstroem2009}) and the standard theory on the transformations between a higher order hyperbolic operator and a first order symmetric hyperbolic system (see, for instance, \cite[\S$2.3$]{Alinhac2009}) in local charts. We omit the details.

	After confirming the above two aspects $(1)$ and $(2)$, we prove the equivalence between the Cauchy problems of the EYM system and the hyperbolic system and then finish this proof. 
\end{proof}


\section{Local well-posedness for hyperbolic systems over tensor bundles}\label{s:local}

\subsection{Symmetric hyperbolic systems over tensor bundles}\label{sec-Model}
Suppose $\mathbb{R} \times\Sigma$ is an $(n+1)$-dimensional Lorentzian manifold with the metric
\begin{equation*}
	\ulg_{ab}=-(dt)_a(dt)_b+\ulh_{ab},
\end{equation*}
and $\Sigma$ is an $n$-dimensional closed Riemannian manifold with the metric $\ulh_{ab}$, and $\nb$ is the Levi-Civita connection of the metric $\ulg_{ab}$ endowed on $[T_0,0) \times \Sigma$. 
The model equation of\footnote{For short, we use $T^{m_k}_{n_k}\Sigma$ to denote $T^{m_k}_{n_k}T\Sigma:=\coprod_{p\in\Sigma}T^{m_k}_{n_k}(T_p\Sigma)$, the tensor bundle of tangent spaces of $\Sigma$. } $u:=(u_{(1)},\cdots,u_{(\ell)})\in \bigoplus^{\ell}_{k=1}T^{m_k}_{n_k}\Sigma$ defined on this time evolutionary manifold $\mathbb{R}  \times \Sigma$ takes the following covariant form,
\begin{align}
	-\mathbf{A}^0(t,u) \nu^c\nb_c u+\mathbf{A}^c(t,u)\tensor{\ulh}{^b_c} \nb_b u={}& G(t,u), \quad &&\text{in }\mathbb{R} \times \Sigma, \label{e:modeq}\\
	u={}&u_0,   \quad &&\text{in }\{t_0\} \times \Sigma.  \label{e:moddt}
\end{align}
We require that the coefficient matrices $\mathbf{A}^0$ and $\mathbf{A}^c \tensor{\ulh}{^b_c}$ are all symmetric (defined in Section \ref{sec-sym}) and $\mathbf{A}^0$ is coercive, i.e., there is some constant $\gamma >0$, such that for all $v \in \bigoplus^{\ell}_{k=1}T^{m_k}_{n_k}\Sigma$,
\begin{align}\label{e:coefcp}
	\gamma \la v,v \ra_{\ulh} \leq \la v, \mathbf{A}^0(t, u)v \ra_{\ulh}.  
\end{align}

\subsubsection{Symmetric linear operators and inner products}\label{sec-sym}
Before stating this model equation, we first introduce some concepts to simplify the statements of the assumptions for the coefficients of the above system \eqref{e:modeq}--\eqref{e:moddt}.

Let \[V:=\bigoplus^{\ell}_{k=1}V_k :=\bigoplus^{\ell}_{k=1}T^{m_k}_{n_k}\Sigma, \quad \text{where} \,\, V_k:=T^{m_k}_{n_k}\Sigma.
\]  
Define the \textit{inner product} of $v:=(v_{(1)},\cdots,v_{(\ell)})\in V$ and $u:=(u_{(1)},\cdots,u_{(\ell)})\in V$,
\begin{align}\label{e:inprod}
	\la v,u \ra_{\ulh} := & \sum^{\ell}_{k=1}\Bigl(\prod_{i=1}^{m_k}  \ulh_{c_{i}b_{i}}\Bigr) \Bigl(\prod_{j=1}^{n_k}  \ulh^{d_{j}a_{j}}\Bigr) \tensor{(v_{(k)})}{^{c_1\cdots c_{m_k}}_{d_1\cdots d_{n_k}}} \tensor{(u_{(k)})}{^{b_1\cdots b_{m_k}}_{a_1\cdots a_{n_k}}}. 
\end{align}
Suppose $\Pbb_{(j)}:V \rightarrow V$ is a \textit{projection} and $\phi_{(j)}: \Ima \Pbb_{(j)} \rightarrow T^{m_j}_{n_j}\Sigma$ an isomorphism satisfying, respectively,
\begin{align} 
	&\Pbb_{(j)}u 
	= (0,\cdots,u_{(j)}, \cdots 0),  \label{e:Pphi1}\\
	&\phi_{(j)} (0,\cdots,u_{(j)}, \cdots 0)= u_{(j)}, \label{e:Pphi2}
\end{align}  
and we denote $\widetilde{\Pbb}_{(j)}:=\phi_{(j)}\circ \Pbb_{(j)}$. Note that $u 
=\sum^\ell_{k=1}\phi^{-1}_{(k)}\widetilde{\Pbb}_{(k)}u$. 

\begin{definition}\label{t:trps}
	Letting a linear map $L(V;V) \ni \mathbf{A}:  V \rightarrow V$, we define the transpose $\mathbf{A}^\star:  V \rightarrow V$ by
	\begin{align*}
		\la v, \mathbf{A} u\ra_{\ulh}=\la \mathbf{A}^\star v, u \ra_{\ulh},
	\end{align*}
	for any $u,v\in V$. 
\end{definition}

Let us denote
$\mathbf{A}_{(kl)}:=  \widetilde{\Pbb}_{(k)} \mathbf{A} \phi_{(l)}^{-1}$ for any $\mathbf{A}\in L(V;V)$ and calculate $\mathbf{A}^\star$,  
\als{
	\la v, \mathbf{A} u \ra_{\ulh}  
	= & \sum_{k=1}^\ell \Bigl(\prod_{i=1}^{m_k}  \ulh_{c_{i}b_{i}}\Bigr) \Bigl(\prod_{j=1}^{n_k}  \ulh^{d_{j}a_{j}}\Bigr) \tensor{(v_{(k)})}{^{c_1\cdots c_{m_k}}_{d_1\cdots d_{n_k}}} \tensor{((\mathbf{A} u)_{(k)})}{^{b_1\cdots b_{m_{k}}}_{a_1\cdots a_{n_{k}}}} \notag  \\ 
	= & \sum_{k=1}^\ell \Bigl(\prod_{i=1}^{m_k}  \ulh_{c_{i}b_{i}}\Bigr) \Bigl(\prod_{j=1}^{n_k}  \ulh^{d_{j}a_{j}}\Bigr) \tensor{(v_{(k)})}{^{c_1\cdots c_{m_k}}_{d_1\cdots d_{n_k}}} \tensor{\bigl(\widetilde{\Pbb}_{(k)} \mathbf{A} \sum_{l=1}^\ell \phi_{(l)}^{-1} u_{(l)} \bigr)}{^{b_1\cdots b_{m_{k}}}_{a_1\cdots a_{n_{k}}}}
	\notag  \\ 
	= &\sum_{k=1}^\ell \sum_{l=1}^\ell \Bigl(\prod_{i=1}^{m_k}  \ulh_{c_{i}b_{i}}\Bigr) \Bigl(\prod_{j=1}^{n_k}  \ulh^{d_{j}a_{j}}\Bigr)    \tensor{(v_{(k)})}{^{c_1\cdots c_{m_k}}_{d_1\cdots d_{n_k}} } \notag \\
	&\hspace{1cm} \cdot \tensor{\bigl(\mathbf{A}_{(kl)} \bigr)}{^{b_1\cdots b_{m_k}  \ha_1\cdots \ha_{n_l} }_{a_1\cdots a_{n_k}  \hb_1\cdots \hb_{m_l}}}   \tensor{ (u_{(l)})}{^{\hb_1\cdots \hb_{m_l}}_{\ha_1\cdots \ha_{n_l}} } \notag  \\
	= & \sum_{l=1}^\ell \Bigl[\sum_{k=1}^\ell \tensor{\bigl(\mathbf{A}_{(kl)} \bigr)}{^{b_1\cdots b_{m_k}  \ha_1\cdots \ha_{n_l}}_{a_1\cdots a_{n_k} \hb_1\cdots \hb_{m_l}} }  \Bigl(\prod_{i=1}^{m_k}  \ulh_{c_{i}b_{i}}\Bigr) \Bigl(\prod_{j=1}^{n_k}  \ulh^{d_{j}a_{j}}\Bigr) \tensor{(v_{(k)})}{^{c_1\cdots c_{m_k}}_{d_1\cdots d_{n_k}}} \notag  \\
	&\hspace{1cm} \cdot \Bigl(\prod_{i=1}^{n_l}  \ulh_{\ha_{i}e_{i}}\Bigr) \Bigl(\prod_{j=1}^{m_l}  \ulh^{\hb_{j}\he_{j}}\Bigr)
	\Bigr]\Bigl(\prod_{i=1}^{n_l}  \ulh^{\hd_{i}e_{i}}\Bigr) \Bigl(\prod_{j=1}^{m_l}  \ulh_{\hc_{j}\he_{j}}\Bigr) \tensor{ (u_{(l)})}{^{\hc_1\cdots \hc_{m_l}}_{\hd_1\cdots \hd_{n_l}} }   \notag  \\
	={} & \la \mathbf{A}^\star v, u \ra_{\ulh}. 
} 
That is
\begin{align}\label{e:Asym}
	& \tensor{ \bigl( (\mathbf{A}^\star)_{(lk)}\bigr)}{^{\he_1\cdots\he_{m_l} d_1\cdots d_{n_k}  }_{e_1\cdots e_{n_l} c_1\cdots c_{m_k}} }= \notag  \\
	& \hspace{1cm} \tensor{ \bigl(\mathbf{A}_{(kl)} \bigr)}{^{b_1\cdots b_{m_k}  \ha_1\cdots \ha_{n_l} }_{a_1\cdots a_{n_k}  \hb_1\cdots \hb_{m_l}} } \Bigl(\prod_{i=1}^{m_k}  \ulh_{c_{i}b_{i}}\Bigr)  \Bigl(\prod_{i=1}^{n_l}  \ulh_{\ha_{i}e_{i}}\Bigr) \Bigl(\prod_{j=1}^{n_k}  \ulh^{d_{j}a_{j}}\Bigr)   \Bigl(\prod_{j=1}^{m_l}  \ulh^{\hb_{j}\he_{j}}\Bigr). 
\end{align}

We call a linear map $\mathbf{A}$ \textit{symmetric} if $\mathbf{A}^\star=\mathbf{A}$. It is equivalent to require $\mathbf{A}^\star_{(lk)}=\mathbf{A}_{(lk)}$.

\subsection{Local existences, uniqueness, continuation principles}

In this section, we establish the local existence and uniqueness theorem, continuation principles of the model equation \eqref{e:modeq}--\eqref{e:moddt}.

\begin{theorem}[Local existence, uniqueness and continuation theorem]\label{t:lcexuqct}
	Suppose $\Sigma$ is a closed manifold, $s\in\Zbb_{\geq3}$, $u_0\in H^s(\Sigma; V)$. We recall that $V:=\bigoplus^{\ell}_{k=1}T^{m_k}_{n_k}\Sigma$ is the tensor bundle.  Then $(a)$ there is a constant $t_\star\in(t_0, + \infty)$ and a unique classical solution, 
	\begin{equation*}
		u\in \bigcap^{1}_{\ell=0}C^\ell([t_0,t_\star),H^{s-\ell}(\Sigma; V)),
	\end{equation*}
	to the equation  \eqref{e:modeq}--\eqref{e:moddt}. 
	
	Moreover, $(b)$ if 
	\begin{equation*}
		\|u\|_{\Li ([t_0,t_\star), W^{1,\infty}(\Sigma;V))}<\infty,
	\end{equation*}
	then the solution $u$ can be uniquely continued, as a classical solution with the same regularity, to a larger time interval $t\in[t_0, t^\star)$ where $t^\star\in(t_\star, + \infty)$.  
\end{theorem}
\begin{remark}\label{t:mtrnt}
	In the proof of this theorem, let us denote $\la [v], [u]\ra_{[[\ulh]]}:=[v]^T[[\ulh]][u]$ the inner product of the vectors $[u]$ and $[v]\in \Rbb^N$ where $[v]^T$ is the transpose of $[v]$ and $[[\ulh]]$ is a matrix of the direct sums of the \textit{Kronecker products} of matrices $[\ulh]$ and $[\ulh^{-1}]$ (we denote $[\ulh]$ the matrix of the metric $\ulh$ in the coordinate $\{x^i\}$ and $\ulh^{-1}$ the inverse of $\ulh$), i.e., 
	\begin{align*}
		[[\ulh]]:=   \bigoplus_{k=1}^\ell \Bigl(\underbrace{[\ulh]\otimes\cdots \otimes[\ulh]}_{n_k}\otimes\underbrace{[\ulh^{-1}]\otimes\cdots \otimes[\ulh^{-1}]}_{m_k}\Bigr)
	\end{align*}
	where $\oplus$ denotes the direct sums of matrices and $\otimes$ the Kronecker products of matrices. Since $\ulh$ is symmetric, $[[\ulh]]$ is a \textit{symmetric} $N\times N$ matrix, where \be\label{def-N} 
	N:= \sum_{k=1}^\ell n^{n_k+m_k}. \ee 
\end{remark}

\begin{proof}
	The proof of this theorem boils down to the local existence and uniqueness theorem of a quasilinear symmetric hyperbolic system, by taking a basis locally in $V$ to evaluate the system  \eqref{e:modeq}--\eqref{e:moddt} at their components. 
	
	\underline{Step $1$. Localization:} For any point $p\in \mathcal{U}\subset \Sigma$ where $\mathcal{U}$ is an open subset of $\Sigma $, there are coordinates $\{x^1,x^2, \cdots, x^n\}$ on $\mathcal{U}$ and define coordinates $\{x^0,x^1,\cdots,x^n\}$ on $\Rbb\times \mathcal{U}$ by letting $x^0=t$. Let us use the notation $[u]$ to denote the vector which collects all the components of $u:=(u_{(1)},\cdots,u_{(\ell)})\in V$ (where $u_{(k)}\in T^{m_k}_{n_k}\Sigma$) in this local coordinate.  
	In specific, according to the coordinate system, the coordinate vector fields $\{\partial/\partial x^i\}$ form a smooth local frame on $\mathcal{U}$, then there is an ordered basis in $V_k:=T^{m_k}_{n_k}\Sigma$ with the \textit{lexicographic order} (also known as the \textit{dictionary order}\footnote{For instance, see \cite[Chapter $1$]{Munkres2015} for the details of the \textit{dictionary order}. In the current case, it means that, for example, $dx^{i_1}\otimes dx^{i_2}\prec dx^{k_1}\otimes dx^{k_2}$ if $i_1<k_1$, or if $i_1=k_1$ and $i_2<k_2$, one can define the general cases by  inductions. }),
	\begin{align*}
		\mathfrak{B}_{(k)}:=\Bigl\{\Bigl(dx^{i_1}\otimes\cdots\otimes dx^{i_{n_k}}\otimes \frac{\partial}{\partial x^{j_{1}}} \otimes\cdots& \otimes
		\frac{\partial}{\partial x^{j_{m_k}}}\Bigr)\;\Big|\;i_1,\cdots,i_{n_k}\in \Zbb, 1\leq i_1,\cdots,i_{n_k}\leq n , \notag  \\
		&j_1,\cdots,j_{m_k}\in \Zbb, 1\leq j_1,\cdots,j_{m_k}\leq n\Bigr\},
	\end{align*}
	for any $k\in \Zbb$ and $1\leq k\leq \ell$. Then we define an \textit{ordered basis} $\mathfrak{B}$ in $V=\bigoplus^{\ell}_{k=1}T^{m_k}_{n_k}\Sigma$ by the \textit{disjoint union} of $\mathfrak{B}_{(k)}$ with the natural order $(k)$, i.e., 
	\begin{equation}\label{e:basis}
		\mathfrak{B}:=\coprod_{k=1}^\ell \mathfrak{B}_{(k)}. 
	\end{equation}
	We list this ordered basis $\mathfrak{B}$ by the $N$-tuple, with $N$ defined in \eqref{def-N} , according to the above order and denote this \textit{row} tuple by $[\mathfrak{B}]$. Correspondingly, we denote the \textit{column} components of $u$ relative to the ordered basis $\mathfrak{B}$ by the $N$-tuple $[u]$, that is,  $u=[\mathfrak{B}][u]$. 
	
	A linear map $\mathbf{A}\in L(V;V)$ can be expressed by a matrix $[\mathbf{A}]$ with the help of the ordered basis $\mathfrak{B}$, i.e., for any $v\in V$, $\mathbf{A} v=[\mathfrak{B}][\mathbf{A}][v]$
	and note
	\begin{align*}
		\partial_b u=(dx^\mu)_b[\mathfrak{B}] \del{\mu}[u]+(dx^\mu)_b[\mathfrak{B}][\underline{\Gamma}]_\mu(t,[u])
	\end{align*}
	where the elements of the matrix $[\underline{\Gamma}]_\mu$ take the form, by denoting $\underline{\Gamma}$ the Christoffel symbols of $\nb$ in each tensor bundle $T^{m_k}_{n_k}\Sigma$, of
	\begin{align*}
		\sum^{m_k}_{r=1}(u_{(k)})^{i_1\cdots p\cdots i_{m_k}}_{j_1\cdots j_{n_k}} \underline{\Gamma}^{i_r}_{\mu p}-\sum^{n_k}_{r=1}(u_{(k)})^{i_1\cdots i_{m_k}}_{j_1\cdots p\cdots j_{n_k}}\underline{\Gamma}^p_{\mu j_r}. 
	\end{align*}
	Since $u_{(k)}=\widetilde{\Pbb}_{(k)} u$ (recall \eqref{e:Pphi1}--\eqref{e:Pphi2}), there is a linear map from the column matrix $[u]$ to its components $(u_{(k)})^{i_1\cdots i_{m_k}}_{j_1\cdots j_{n_k}}$.

	Consequently, the model equation \eqref{e:modeq} can be rewritten in terms of the vector $[u]\in \Rbb^{N}$ where we recall \eqref{def-N} for $N$, i.e., under the ordered basis $[\mathfrak{B}]$, the equation \eqref{e:modeq} becomes 
	\begin{equation}\label{e:modmtr1}
		[\mathbf{A}^0](t,x,[u])\del{t} [u]+[\mathbf{A}^i](t,x,[u])\del{i} [u]= [\tilde{G}](t,x,[u]),  
	\end{equation}
	where $[\tilde{G}]=-[\mathbf{A}^i][\underline{\Gamma}]_i-[\mathbf{A}^0][\underline{\Gamma}]_0+[G]$ has included the Christoffel symbol terms of the connection $\nb$. Since $\ulh$ is positive definite, then so is the matrix $[[\ulh]]$ on $[T_0,0]\times \mathcal{U}$. Multiplying $[[\ulh]]$ on both sides of \eqref{e:modmtr1}, we have
	\begin{align}\label{e:modmtr2}
		[[\ulh]][\mathbf{A}^0](t,x,[u])\del{t} [u]+[[\ulh]][\mathbf{A}^i](t,x,[u])\del{i} [u]= [[\ulh]][\tilde{G}](t,x,[u]).	
	\end{align}
	
	We \textit{claim} that \eqref{e:modmtr2} is a \textit{symmetric hyperbolic} system of the vector $[u]$ on $[t_0, + \infty)\times \mathcal{U}\subset [t_0, +\infty)\times \Rbb^N$ which is covered by the coordinate system $\{x^i\}$ since $[[\ulh]][\mathbf{A}^0]$ and $[[\ulh]][\mathbf{A}]$ are symmetric and $[[\ulh]][\mathbf{A}^0]\gtrsim \mathds{1}_{N\times N}>0$, where $\mathds{1}_{N\times N}$ is an identity matrix, due to the fact that $\mathbf{A}^0$ and $\mathbf{A}$ are symmetric and \eqref{e:coefcp}. Let us prove this claim.  
	Note firstly $\la v,u \ra_{\ulh}=\la [v], [u]\ra_{[[\ulh]]}=[v]^T[[\ulh]][u]$ by \eqref{e:inprod} and $[[\ulh]]$ is symmetric. If $\mathbf{A}$ is symmetric, then $\la v, \mathbf{A} u\ra_{\ulh}=\la \mathbf{A} v, u \ra_{\ulh}$ implies 
	\begin{equation*}
		[v]^T[[\ulh]][\mathbf{A}][u]=\la [v], [\mathbf{A}][u]\ra_{[[\ulh]]}=\la [\mathbf{A}] [v], [u] \ra_{[[\ulh]]}=[v]^T[\mathbf{A}]^T[[\ulh]][u]=[v]^T([[\ulh]][\mathbf{A}])^T[u], 
	\end{equation*}
	for all $[v], [u]\in \Rbb^N$, then $[[\ulh]][\mathbf{A}]=([[\ulh]][\mathbf{A}])^T$, that is, $[[\ulh]][\mathbf{A}]$ is symmetric. The \eqref{e:coefcp}, and the positive definiteness of $[[\ulh]]$ imply that, for any $[u]\in \Rbb^N$, there is a small constant $\tilde{\kappa}$, such that   
	\begin{equation*}
		\frac{\tilde{\kappa}}{\gamma_1} [u]^T [u]\leq \frac{1}{\gamma_1} [u]^T[[\ulh]][u]=\frac{1}{\gamma_1} \la u,u \ra_{\ulh} \leq \la u, \mathbf{A}^0 u \ra_{\ulh} =[u]^T [[\ulh]] [\mathbf{A}^0] [u],
	\end{equation*}
	that is, the matrix $[[\ulh]] [\mathbf{A}^0]$ is positive definite. Therefore, the system \eqref{e:modmtr2} is a symmetric hyperbolic equation of vector $[u]$ defined on $[t_0, + \infty) \times \mathcal{U}\subset [t_0, + \infty)\times \Rbb^N$. 
	
	Let $\psi\in C^\infty_0(\mathcal{U})$ be a cut-off function satisfying $\psi(q)=1$ for every $q\in\overline{\mathcal{V}}$ where $\mathcal{V}$ is an open subset obeying that $\overline{\mathcal{V}}\subset \mathcal{U}$ is compact.  Modifying the initial data $[u_0]$ (see \eqref{e:moddt}), the column matrix form of the initial data $u_0$, by multiplying all its components with $\psi$, we obtain the localized initial data 
	\begin{equation}\label{e:moddt2}
		[u]=\psi [u_0], \quad \text{in } \{t_0\}\times \mathcal{U}, 
	\end{equation}
	
	Gathering \eqref{e:modmtr2} and \eqref{e:moddt2} together, and using the theory of the standard hyperbolic equations (see, for instance,  \cite{Taylor2010,Majda2012,Gavage2007a}), by \textit{trivially}\footnote{That means the extended solution becomes trivial outside a larger domain $\mathcal{V}^\prime\supset\mathcal{V}$ where $\overline{\mathcal{V}^\prime}\subset \Rbb^n$ is compact. } extending the initial data \eqref{e:moddt2} to $\{t_0\} \times \Rbb^{n}$,  there is a constant $t_1>t_0$ and a \textit{unique} solution $[u]$ solves the resulting system \eqref{e:modmtr2}--\eqref{e:moddt2} satisfying 
	\begin{equation*}
		[u]\in \bigcap^{1}_{\ell=0}C^\ell([t_0, t_1), H^{s-\ell}(\Rbb^n; \Rbb^N)).
	\end{equation*}
	In addition, there is a lens-like domain (the \textit{domain of determination of $\mathcal{V}$}, see the local energy estimates, for example, in \cite[\S $3.1$]{Liu2021} and \cite[\S $2.3$]{Majda2012}) $\mathcal{W}\subset[t_0, t_1] \times \mathcal{V}$ where the solution $[u]$ defined on $\mathcal{W}$ is determined only by the data $[u_0]$ on $\{t_0\}\times \mathcal{V}$.

	\underline{Step $2$. Existence on $[t_0, t_\star) \times\Sigma$:} 
	As we stated above, for every point $p\in \Sigma$, there are open subsets $\mathcal{V}$ and $\mathcal{U}$ satisfying the above requirements ($\overline{\mathcal{V}}\subset \mathcal{U}$) and let us denote them by $\mathcal{V}_p$ and $\mathcal{U}_p$, in this step, to distinguish different points $p \in \Sigma$. Since $\Sigma$ is a compact manifold, there are a finite number (denoted by $m$) of subsets $\mathcal{V}_p$ to cover $\Sigma$, i.e., $\Sigma\subset \cup^m_{l=1}\mathcal{V}_{p_l}$. From Step $1$, for every open subset $\mathcal{V}_{p_l}$, there is a time $t_l<0$ and a unique solution $[u]\in \bigcap^{1}_{\ell=0}C^\ell([t_0, t_l),H^{s-\ell}(\mathcal{V}_{p_l}; \Rbb^N))$. Correspondingly, there are domains of determination $\mathcal{W}_{p_l}$ of $\mathcal{V}_{p_l}$ for every $p_l$. By taking $t_\star<\min_{1\leq l\leq m} \{t_l\}$ small enough, it ensures for every point $q\in [t_0, t_\star) \times \Sigma$, that there is a point $p_l\in\mathcal{V}_{p_l}$ with the local coordinate $\{x^i\}$, such that $q\in \mathcal{W}_{p_l}$. Then, we have constructed a solution $u$ in  $[t_0, t_\star) \times \Sigma$ to the equations \eqref{e:modeq}--\eqref{e:moddt}, i.e., there is a tensorial function $u\in V$ defined in  $[t_0, t_\star) \times \Sigma$, such that  \eqref{e:modeq}--\eqref{e:moddt} hold. This is because for every point $q \in [t_0, t_\star) \times \Sigma$, there is a point $p_l\in\mathcal{V}_{p_l}$ with the arbitrary local coordinate $\{x^i\}$, such that $q\in \mathcal{W}_{p_l}$, and in this coordinate $\{x^\mu\}$, $[u]$ solves \eqref{e:modmtr2}, further \eqref{e:modmtr1}, i.e., the column matrix $[u]$ in any  coordinate $\{x^\mu\}$ ensures that the matrix equation \eqref{e:modmtr1} holds. 
	
	\underline{Step $3$. Uniqueness on $[t_0, t_\star) \times \Sigma$:} The uniqueness can be concluded by contradictions. We briefly state as below. If for every point $q\in[t_0, t_\star) \times \Sigma$, there are two solutions $u_1,u_2\in V$ and $u_1\neq u_2$ in a neighborhood of $q$, such that  \eqref{e:modeq} holds under the same initial data \eqref{e:moddt}. Then evaluating $u_1$ and $u_2$ at their components under a local chart near $q$, we obtain two different column matrices $[u_1]\neq [u_2]$ under the same ordered coordinate basis, which contradicts with the fact that there is a unique solution of \eqref{e:modmtr1}. In addition, similarly, we claim that if $\mathcal{W}_{p}\cap \mathcal{W}_{p^\prime}\neq \emptyset$, then the two solutions located in $\mathcal{W}_{p}$ and $\mathcal{W}_{p^\prime}$ respectively agree with each other on the overlap $\mathcal{W}_{p}\cap \mathcal{W}_{p^\prime}$. Let us denote the coordinate in $\mathcal{W}_p$ by $\{x^\mu\}$, and the one in $\mathcal{W}_{p^\prime}$ by $\{y^\mu\}$. The components of $u$ aligned in the column matrix in $\{x^\mu\}$ are denoted by $[u]_x$, and $[u]_y$ in $\{y^\mu\}$, respectively. We conclude that if we write $[u]_y$ in the coordinate $\{x^\mu\}$ after a coordinate transform, it becomes $[u]_x$, otherwise, it contradicts with the uniqueness of the solution of \eqref{e:modmtr1} as well.

	\underline{Step $4$. Continuation principle:} 
	Let $u$ be the local solution to \eqref{e:modeq}--\eqref{e:moddt} on $[t_0,t_\star) \times \Sigma$ that, by the assumption of this theorem, admits a uniform bound $\|u\|_{\Li ([t_0,t_\star), W^{1,\infty}(\Sigma;V))}<\infty$. Then, as the proofs of the standard continuation principle, for any $t \in (t_0, t_\star)$ close to $t_\star$, we can construct a unique solution as Step $1$--$3$ that exists at least on $[t, t + \epsilon) \times \Sigma$ for a constant $\epsilon$ independent of $t$. Letting $t_\star-\epsilon^\prime$ be the new initial time with $\epsilon^\prime < \epsilon$ small enough, we can continue the solution to a larger time interval $[t_0, t^\star),\, t^\star > t_\star$. We omit the details and then complete the proof of this theorem. 
\end{proof}


\section{Proofs of the Main Theorem \ref{t:mainthm}}\label{s:mainpf}
This section is dedicated to the proof of the main Theorem \ref{t:mainthm}. Let the state space\footnote{See \cite{Majda2012} and it is the natural domain of the physical quantities. } of the solution of this system satisfy the condition $| \mathbf{U}|<C_0$ ($C_0>0$ is a given constant) such that the conformal metric $g^{ab}$ remains non-degenerate. 
The proof follows immediately as a consequence of \textit{Theorem \ref{t:lcexuqct}} and \textit{Theorem \ref{t:depctrnt}}:  

$(1)$ Verify that the main equation \eqref{e:EYMmain} is a tensorial symmetric hyperbolic system, which will be deferred for the moment;   

$(2)$ Apply the local existence Theorem \ref{t:lcexuqct} of the model equation to obtain the local existence, uniqueness and the continuation principle of the solution to the main equation \eqref{e:EYMmain};

$(3)$ By Theorem \ref{t:depctrnt}, the equivalence between the hyperbolic system \eqref{e:EYMmain} and the EYM system, we conclude the local existence and the continuation principle of the EYM equations, i.e., we prove the main Theorem \ref{t:mainthm}.

In the end, let us show that the main equation \eqref{e:EYMmain} is a tensorial symmetric hyperbolic system to confirm the Step $(1)$ above.

	\underline{$(a)$ The symmetric  coefficients $\mathbf{A}^{c}$:}		
	According to Definition \ref{t:trps} and \eqref{e:Asym}, it suffices to check that each block matrix of $\mathbf{A}^{ c}$ is symmetric. We only list some cases here, and the others can be verified in the same way. 
	\begin{gather*}
		(\check{\mathbf{A}}^c_{(12)})^{\he} \ulh_{\he\bar{e}} = -g^{dc}\ulh_{d \hc}=(\check{\mathbf{A}}^c_{(21)})_{\hc},  \\
		\tensor{\bigl( \tilde{\mathbf{A}}^{ c}_{(22)}\bigr) }{^{\check{d}}_f }  \ulh_{\check{d} \hc} \ulh^{f \he}   = - T^c h^{\hd d}\ulh_{d f} \tensor{\ulh}{^{\check{d}}_{\hd}} \ulh_{\check{d} \hc} \ulh^{f \he} = - T^{ c} h^{\hd d} \tensor{\ulh}{^{\he}_{d}} \ulh_{\hd\hc} = \tensor{\bigl( \tilde{\mathbf{A}}^{ c}_{(2 2)}\bigr) }{^{ \he}_{\hc} }, \\ 
		\tensor{\bigl( \tilde{\mathbf{A}}^{ c}_{(31)}\bigr) }{_{ \bar a \bar b \check{e}} } \ulh^{\check{a} \bar a} \ulh^{\check{b} \bar b} \ulh^{b \check{e}}  = h^{a c} \tensor{h}{^b_{e}} \ulh_{a \bar a} \ulh_{b \bar b} \tensor{\ulh}{^{e}_{\check{e}}}\ulh^{\check{a} \bar a} \ulh^{\check{b} \bar b} \ulh^{b \check{e}}=   h^{a c} \tensor{h}{^b_{e}} \tensor{\ulh}{^{\check{a}}_a} \tensor{\ulh}{^{\check{b}}_b}   \ulh^{b e} = \tensor{\bigl( \tilde{\mathbf{A}}^{ c}_{(1 3)}\bigr) }{^{\check{a}\check{b}b} }, \\
		\tensor{\bigl( \tilde{\mathbf{A}}^{ c}_{(32)}\bigr) }{^{\check{d}}_{ \bar a \bar b} } \ulh^{\ha \bar a} \ulh^{\hb \bar b} \ulh_{\check{d} f}  = h^{a \hd} h^{b c}\ulh_{a \bar a} \ulh_{b \bar b} \tensor{\ulh}{^{\check{d}}_{\hd}}  \ulh^{\ha \bar a} \ulh^{\hb \bar b} \ulh_{\check{d} f} =  h^{a \hd} h^{b c} \tensor{\ulh}{^{\ha}_a} \tensor{\ulh}{^{\hb}_b}   \ulh_{\hd f} 
		= \tensor{\bigl( \tilde{\mathbf{A}}^{ c}_{(2 3)}\bigr) }{^{\ha \hb}_f }.
	\end{gather*}

	\underline{$(b)$ The positive definiteness  of $\mathbf{A}^0$:}
	Due to  
	$T^c\nu_c=-(-\lambda)^{ \frac{1}{2}} $ and $\tensor{h}{^c_a} \nu_c=0$, it is direct, by calculations, to verify that $\check{\mathbf{A}}^0=\diag\{-\nu_e g^{ec} \nu_c, \, \ulh_{e\hc} g^{ed} \tensor{\ulh}{^{\he}_d}, \, 1\}$ and 
	\begin{align*}
		\tilde{\mathbf{A}}^c\nu_c=\p{(-\lambda)^{ \frac{1}{2}} h_{\he e} \ulh^{b\he} \tensor{\ulh}{^e_{\check{e}}}  & 0 & 0 & 0 \\
			0 & (-\lambda)^{ \frac{1}{2}} h^{\hd d} \ulh_{df} \tensor{\ulh}{^d_{\check{d}}}  & 0 & 0 \\
			0 & 0 & (-\lambda)^{ \frac{1}{2}}h^{\ha a} h^{\hb b} \ulh_{a\bar{a}}\ulh_{b\bar{b}} \tensor{\ulh}{^{\check{a}}_{\ha}} \tensor{\ulh}{^{\check{b}}_{\hb}} & 0 \\
			0 & 0 & 0 & (-\lambda)^{ \frac{1}{2}}h^{r\hc} \ulh_{\check{c}r} \tensor{\ulh}{^{\check{s}}_{\hc}}
		} 
	\end{align*}
	are both positive definite\footnote{They are \textit{positive definite} by viewing that the \textit{domain} of the operator $\mathbf{A}^0$ is a subspace of spacetime tensors $\bigoplus^{\ell}_{k=1}T^{m_k}_{n_k}\mathcal{M}$  which is naturally isometric to the spatial tensors  $\bigoplus^{\ell}_{k=1}T^{m_k}_{n_k}\Sigma$. }.


	\appendix	
	
	\section{Some calculations}\label{s:App1}
	\subsection{Useful calculations} In this appendix, we collect plenty calculations that are used throughout the paper to simplify the expositions.
	
	\begin{lemma}\label{t:Rdop}
		The Ricci tensors can be expressed in terms of $\udl{\nabla}_a$ and $\tensor{\udl{R}}{_{cde}^a}$,
		\al{RICCI}{
			R^{ab}=\frac{1}{2} g^{cd}\udn{c}\udn{d}g^{ab}+\nabla^{(a}X^{b)}+\udl{R}^{ab}+P^{ab}(g^{-1})+Q^{ab}(g^{-1}, \udn{} g^{-1}),
		}
		where $X^a$ is given by \eqref{def-X}, and
		\als{
			P^{ab}(g^{-1})
			=& -\frac{1}{2} (g^{ac}-\ulg^{ac})\ulg^{de}\tensor{\udl{R}}{_{cde}^b} -\frac{1}{2} \ulg^{ac}(g^{de}-\ulg^{de})\tensor{\udl{R}}{_{cde}^b} -\frac{1}{2} (g^{ac}-\ulg^{ac})(g^{de}-\ulg^{de})\tensor{\udl{R}}{_{cde}^b}   \\
			&-\frac{1}{2}(g^{bc}-\ulg^{bc})\ulg^{de}\tensor{\udl{R}}{_{cde}^a}-\frac{1}{2}\ulg^{bc}(g^{de}-\ulg^{de})\tensor{\udl{R}}{_{cde}^a}-\frac{1}{2}(g^{bc}-\ulg^{bc})(g^{de}-\ulg^{de})\tensor{\udl{R}}{_{cde}^a},
		}
		and
		\als{
			& Q^{ab}(g^{-1}, \udn{} g^{-1})  \\
			={}& -\frac{1}{4}\bigl(g^{ad}g^{bf}\udn{d}g_{ce}\udn{f}g^{ce}+g^{ad}g^{bf}\udn{f}g_{ce}\udn{d}g^{ce}+g_{ef}g^{ac}\udn{c}g^{bd}\udn{d}g^{ef}  +g_{ef}g^{bd}\udn{d}g^{ac}\udn{c}g^{ef}\bigr) \\
			& +\frac{1}{2}\big(g^{ae}g_{fc}\udn{e} g^{bd} \udn{d}g^{fc}+g^{ae} g^{bd} \udn{e} g_{fc} \udn{d} g^{fc}-\udn{c}g^{ae}\udn{e}g^{cb}-\udn{e}g^{bc}\udn{c} g^{ea}\big) -g^{ac}\tensor{X}{^b_{cd}}X^d  \\
			& +g^{af}g^{bd}\tensor{X}{^{c}_{fd}}\tensor{X}{^e _{ce}}-g^{af}g^{bd}\tensor{X}{^{c}_{ed}}\tensor{X}{^e _{cf}}+\tensor{X}{^e_{ed}} g^{af} \udn{f} g^{bd} -\tensor{X}{^e_{fd}}g^{af}\udn{e} g^{bd}-\tensor{X}{^e_{fd}}g^{bd}\udn{e} g^{af}.
		}
	\end{lemma}
	
	\begin{remark}
		With the above expressions of $P^{ab}$ and $Q^{ab}$, it is obvious that $P^{ab}$ depends on $g^{-1}- \ulg^{-1}$  (up to quadratic) and $Q^{ab}$ is quadratic in $\udn{}g^{-1}$.
	\end{remark}
	
	\begin{proof}
		The proof of \eqref{E:RICCI} follows by a straightforward calculation and noticing the following identities
		\als{
			g^{af}g^{bd}\udn{c}g_{fd}=-\udn{c}g^{ab} \AND -\frac{1}{2}(\ulg^{ac}\ulg^{de}\tensor{\udl{R}}{_{cde}^b}+\ulg^{bc}\ulg^{de}\tensor{\udl{R}}{_{cde}^a}) \equiv \udl{R}^{ab}.
		}
		We omit the lengthy details.
	\end{proof}

	Direct calculations yield the next proposition. 
	\begin{lemma}\label{lem-identity}
		Suppose $Q^{edc}$ is the symmetrizing tensor given by \eqref{def-Q}, the projection $\tensor{h}{^a_b}$, induced metric $h_{ab}$ and $T^a$ are defined in \S\ref{s:geohysf},
		then the following identities hold,
		\begin{align*}
			Q^{edc} T_c ={}&   h^{e d} + T^e T^d,  \nnb \\
			Q^{edc} \tensor{h}{^a_d} ={}&  T^e h^{c a} - T^c h^{e a} \nnb\\
			Q^{edc} T_d ={}&  - h^{e c} + T^c T^{e}, 
		\end{align*}
		and $Q^{edc} T_d T_c =  - T^e$, $Q^{edc} \tensor{h}{^a_d} T_c =  h^{e a}$, $Q^{edc}T_d \tensor{h}{^a_c} =  - h^{a e}$, and
		\gat{
			T_e Q^{edc}T_d =  - T^c,  \quad 
			\tensor{h}{^b_e} Q^{edc} \tensor{h}{^a_d} =   - T^c h^{a b} \quad
			T_e Q^{edc} \tensor{h}{^a_d} =  - h^{a c}, \nnb \\ 
			T_e Q^{edc} T_d T_c= 1,  \quad  T_e Q^{edc} T_d \tensor{h}{^a_c} = 0, \quad  
			T_e Q^{edc} \tensor{h}{^a_d} T_c=0, \nnb \\  T_e Q^{edc} \tensor{h}{^a_d} \tensor{h}{^b_c} = - h^{a b}, \quad
			\tensor{h}{^a_e} Q^{edc} \tensor{h}{^b_d} T_c = h^{a b}, \quad \tensor{h}{^a_e} Q^{edc} \tensor{h}{^b_d} \tensor{h}{^f_c} = 0. \nnb
		}
		In addition, the symmetry of $Q^{e d c}$ leads to
		\gat{
			\tensor{h}{^a_e} Q^{edc}T_d = - h^{a c}, \quad \tensor{h}{^{b}_e} Q^{edc} T_d \tensor{h}{^a_c} =  - h^{a b}, \quad  \tensor{h}{^{b}_e} Q^{edc} T_d T_c = 0. \nnb
		}
		Analogous identities hold for $\underline{Q}^{edc}$.	
	\end{lemma}

	The proofs of Lemma \ref{t:rdein} and Theorem \ref{thm-hyperb-gravity} are presented as below.
	
	\subsection{Proof of Lemma \ref{t:rdein}}	
	\begin{proof}
		Firstly, by taking the trace of the Einstein equations $R^{ab}-\frac{1}{2}R g^{ab}+\Lambda g^{ab}= \mathcal{T}^{ab}$, we obtain $R=\frac{2}{n-1}[(n+1)\Lambda- \mathcal{T}]$ with $\mathcal{T} := g^{a b} \mathcal{T}_{a b}$. It, inserting $R$ into the Einstein equations,  further yields
		\begin{equation*}
			R^{ab}-\frac{2}{n-1}\Lambda g^{ab}= \mathcal{T}^{ab}-\frac{1}{n-1} \mathcal{T} g^{ab}.
		\end{equation*}
		The Ricci tensors, by Lemma \ref{t:Rdop}, can be expressed 
		\als{
			R^{ab}=\frac{1}{2} g^{cd}\udn{c}\udn{d}g^{ab}+\nabla^{(a}X^{b)}+\udl{R}^{ab}+P^{ab}(g^{-1})+Q^{ab}(g^{-1}, \udn{} g^{-1}),
		}
		Substitute the wave gauge \eqref{E:CONSTR1} into the above formula and then we arrive at
		\als{
			& g^{cd}\udn{c}\udn{d} g^{ab}+2\nabla^{(a} f^{b)}+2\ulR^{ab}+2P^{ab}+2Q^{ab} -\frac{4\Lambda}{n-1}  g^{ab} =2 \mathcal{T}^{ab}-\frac{2}{n-1} \mathcal{T} g^{ab} \notag  \\
			&\hspace{0.5cm} = {}  2 g^{bd}F^{ac}\cdot F_{dc}-\frac{1}{n-1}g^{ab}F^{cd}\cdot F_{cd} \nnb \\
			&\hspace{0.5cm}= {}  2 g^{bd} (-E^a E_d + H^{a b} H_{d b} +  T_d H^{a b} E_b + T^a E^b H_{d b} + T^a T^d |E|^2 ) -\frac{1}{n-1} g^{ab} (-2 |E|^2 + |H|^2).
		} 
		This completes the proof.  
	\end{proof}

	\subsection{Proof of Theorem \ref{thm-hyperb-gravity}}

	\begin{proof}
		For the reduced Einstein equations \eqref{eq-redu-Einstein}, by denoting $\tensor{g}{^{ab}_d}:=\nb_d g^{ab}$, we note the identities
		\begin{gather*}
			\underline{Q}^{edc}\nb_c\nb_d g^{ab} = \nu^e g^{dc}\nb_c \tensor{g}{^{ab}_d}+\nu^d g^{ec}(\nb_c\nb_d-\nb_d\nb_c) g^{ab}= \nu^e g^{dc}\nb_c \tensor{g}{^{ab}_d}+\nu^d g^{ec}(\tensor{\underline{R}}{_{cdf}^a} g^{fb}+\tensor{\underline{R}}{_{cdf}^b}g^{fa})\\
			\intertext{and}
			\nu^c\nb_c g^{ab}= \tensor{g}{^{ab}_c} \nu^c. 
		\end{gather*}
		Direct calculations conclude this theorem.  
	\end{proof}

	\subsection{Leibniz rule for Lie derivative over principle bundle}\label{s:lieleb}	
	\begin{lemma}\label{t:lieleb}
		Suppose $G$ is a Lie group and $\mathcal{G}$ is the associated Lie algebra, and there is a $G$-principle bundle over the base manifold $\mathcal{M}$. Let $V$ be any vector field on $\mathcal{M}$, $A_a$ a $\mathcal{G}$-valued one form on $\mathcal{M}$. Then
		\begin{equation*}
			\lie_{V} [A_a, A_b]=[\lie_{V} A_a,A_b]+[A_a,\lie_V A_b].
		\end{equation*}
	\end{lemma}
	\begin{proof}
		Let $\{\theta_\mu\}$ be a basis of a  faithful real matrix representation of the Lie algebra $\mathcal{G}$,
		\begin{align*}
			A_a=A_a^\mu\theta_\mu \AND A_b=A_b^\nu\theta_\nu.
		\end{align*}
		Then, since $[\theta_\mu,\theta_\nu]$ are determined by structural constants of $\mathcal{G}$, we calculate $\lie_{V} [A_a, A_b]$, with the help of the linearity and the Leibniz's rule of Lie derivative,
		\begin{align*}
			&\lie_{V} [A_a, A_b]=\lie_V([\theta_\mu,\theta_\nu] A^\mu_a  A^\nu_b)=(\lie_V A^\mu_a)  A^\mu_b [\theta_\mu,\theta_\nu]+A^\mu_a  \lie_V A^\nu_b [\theta_\mu,\theta_\nu] \notag  \\
			&\hspace{1cm}= [\lie_V A^\mu_a \theta_\mu, A^\nu_b \theta_\nu] + [ A^\mu_a \theta_\mu, \lie_V A^\nu_b \theta_\nu] =[\lie_V A_a , A_b ] + [ A_a , \lie_V A_b ].
		\end{align*}
		We then complete the proof.
	\end{proof}

	\section*{Acknowledgement}
	C.L. is supported by the Fundamental Research Funds for the Central Universities, HUST: $5003011036$. 
	J.W. is supported by NSFC (Grant No. 11701482).

\end{document}